\documentclass[dvips,10pt]{article}

\newcommand{\bs}{\boldsymbol}
\newcommand{\bsu}{ {\bs{u}}}
\newcommand{\bsf}{ {\bs{f}}}
\newcommand{\bsx}{ {\bs{x}}}
\newcommand{\bsz}{ {\bs{z}}}
\newcommand{\bsa}{ {\bs{a}}}
\newcommand{\bsy}{ {\bs{y}}}
\newcommand{\bsX}{ {\bs{X}}}
\newcommand{\bsY}{ {\bs{Y}}}
\newcommand{\bsZ}{ {\bs{Z}}}
\newcommand{\bsv}{ {\bs{v}}}

\newcommand{\bsF}{ {\bs{F}}}
\newcommand{\bsG}{ {\bs{G}}}
\newcommand{\bsM}{ {\bs{M}}}

\usepackage{amssymb,amsbsy,amsmath,amsthm}
\usepackage{tcolorbox}
\usepackage{epsf,epsfig}
\usepackage{texdraw}
\usepackage{bbding,pifont}

\newtheorem{proposition}{Proposition}

\newtheorem*{criterion*}{Criterion}
\newtheorem*{definition*}{Definition}
\newtheorem{definition}{Definition}

\theoremstyle{remark} 
\newtheorem{remark}{Remark}

\title{Noncommutative discrete equations, symmetries and reductions}
\author{Pavlos Xenitidis\footnote{Email: xenitip@hope.ac.uk}\\ School of Computer Science and the Environment\\ Liverpool Hope University, L16 9JD  Liverpool, UK}
\date{\today}

\begin{document}
	
	\maketitle
	\begin{abstract}
	Employing the Lax pairs of the noncommutative discrete potential Korteweg--de Vries (KdV) and Hirota's KdV equations, we derive differential--difference equations that are consistent with these systems and serve as their generalised symmetries. Miura transformations mapping these equations to a noncommutative modified Volterra equation and its master symmetry are constructed. We demonstrate the use of these symmetries to reduce the potential KdV equation, leading to a noncommutative discrete Painlev{\`{e}} equation and to a system of partial differential equations that generalises the Ernst equation and the Neugebauer--Kramer involution. Additionally, we present a Darboux transformation and an auto-B\"acklund transformation for the Hirota KdV equation, and establish their connection with the noncommutative Yang--Baxter map $F_{III}$.

		
		{\bf{Classification MSC:}} 39A36, 39A70, 39B22, 37C79, 37K60

	\end{abstract}

	\section{Introduction}
		
	The notion of integrability for (systems of) differential, differential--difference, and difference equations is well established and incorporates various aspects, such as Lax representations, Miura and B\"acklund transformations, soliton solutions, and infinite hierarchies of symmetries and conservation laws, to name a few. Connections among different types of integrable systems are also well known; for example, evolutionary differential--difference equations serve as symmetries of difference equations, while difference equations may express the superposition principle for B\"acklund transformations of differential equations. Interrelations among different aspects of integrability are also well established. For instance, a Lax pair can be used to construct symmetries and auto-B\"acklund transformations for the corresponding integrable system, whereas symmetries may provide a link between continuous and discrete integrable systems. 
	
	Various methods and approaches exist for establishing the integrability of systems with commutative variables. In the case of the discrete systems, it is well known that multidimensional consistency provides a powerful tool for constructing both a Lax pair and an auto-B\"acklund transformation; see, for example, \cite{HJN} and references therein. Symmetries can be derived systematically either from first principles (see \cite{FX,TTX2}) or through the theory of integrability conditions (see \cite{BX,LY,MWX,MX,X}). Many of these structures and methods naturally extend to the noncommutative setting. In some cases, this extension is straightforward; in others, it requires more intricate constructions and computations. 
		
	In this paper, we derive generalised symmetries of noncommutative difference equations defined on an elementary quadrilateral of the ${\mathbb{Z}}^2$ lattice. Our approach follows that of \cite{FX}, which systematically derives symmetries using the Lax pair of the underlying discrete system. This method extends naturally to noncommutative settings and offers a broadly applicable framework. To demonstrate its effectiveness, we consider two illustrative examples. We begin with the discrete potential KdV equation \cite{FNC,Betal}
	\begin{equation*}
		\left({\bs{u}}_{0,0} - {\bs{u}}_{1,1}\right) \left( {\bs{u}}_{1,0}-{\bs{u}}_{0,1}\right) = \alpha - \beta,
	\end{equation*}
	and its corresponding Lax pair, from which we derive the lowest order generalised symmetries. These symmetries are non-polynomial differential–difference equations,
	$$	\frac{{\rm{d}} {\bs{u}}_{0,0}}{{\rm{d}} t_1} =  \left({\bs{u}}_{1,0}- {\bs{u}}_{-1,0}\right)^{-1}, ~~~	\frac{{\rm{d}} {\bs{u}}_{0,0}}{{\rm{d}} x} =  n \left({\bs{u}}_{1,0}- {\bs{u}}_{-1,0}\right)^{-1}, ~~~ 	\frac{{\rm{d}} \alpha}{{\rm{d}} x} = -1, $$
	that generalise those of the commutative case \cite{TTX2}.
	
	We employ these symmetries in two distinct settings to reduce the discrete potential KdV equation. The first reduction yields a three-dimensional map related to the difference equation
	\begin{equation*} 
		r f_{n+1} \left( \bsY_{n+1} \bsY_n - r \right)^{-1} + r f_{n} \left( \bsY_{n-1} \bsY_n - r \right)^{-1}  + r \lambda_{n+m} \bsY_n^{-1} + \lambda_{n+m+1} \bsY_n + g_m + f_{n+1} = 0,
	\end{equation*}
	which can be viewed as a noncommutative discrete Painlev\`e II equation. 
	
	In the second context, we establish a correspondence between a class of solutions of the discrete equation and solutions of the integrable system of partial differential equations
	\begin{align*}
	\frac{\partial {\bs{u}}_{1}}{\partial \beta} &= \frac{1}{\alpha-\beta} \left({\bs{u}}_{1} -{\bs{u}}_{2}\right) \left(m - \frac{\partial {\bs{u}}}{\partial \beta}  \left( {\bs{u}}_{1} -{\bs{u}}_{2} \right)\right), \nonumber \\
	\frac{\partial {\bs{u}}_{2}}{\partial \alpha} &=\frac{1}{\alpha-\beta} \left({\bs{u}}_{1} -{\bs{u}}_{2}\right) \left(n+ \frac{\partial {\bs{u}}}{\partial \alpha}  \left( {\bs{u}}_{1} -{\bs{u}}_{2} \right)\right),  \\
	\frac{\partial^2 {\bs{u}}}{\partial \alpha \partial \beta} &= \frac{1}{\alpha-\beta} \left(  \frac{\partial {\bs{u}}}{\partial \alpha} \left({\bs{u}}_{1} -{\bs{u}}_{2}\right)  \frac{\partial {\bs{u}}}{\partial \beta} + \frac{\partial {\bs{u}}}{\partial \beta} \left({\bs{u}}_{1} -{\bs{u}}_{2}\right)  \frac{\partial {\bs{u}}}{\partial \alpha} + n \frac{\partial {\bs{u}}}{\partial \beta} -  m \frac{\partial {\bs{u}}}{\partial \alpha} \right). \nonumber
	\end{align*}
	The integrability of this noncommutative system is established by the existence of a Lax pair and an auto-B\"acklund transformation. Moreover, it incorporates and generalises well-known systems from Mathematical Physics, including the Ernst equation and the Neugebauer–Kramer involution.
	
	Our second illustrative example is Hirota's discrete KdV equation \cite{Betal},
	\begin{equation*} 
		{\bs{u}}_{0,0} + \alpha {\bs{u}}_{1,0}^{-1} - \alpha {\bs{u}}_{0,1}^{-1}-{\bs{u}}_{1,1} = {\bs{0}}.
	\end{equation*}
	Using its associated Lax pair, we derive the corresponding generalised symmetries. 
	$$\frac{{\rm{d}} {\bs{u}}_{0,0}}{{\rm{d}} t_1} =  {\bs{u}}_{0,0} {\bs{f}}_{1,0} -{\bs{f}}_{0,0}  {\bs{u}}_{0,0},~~~ \frac{{\rm{d}} {\bs{u}}_{0,0}}{{\rm{d}} x} =  n {\bs{u}}_{0,0} {\bs{f}}_{1,0} -(n-1) {\bs{f}}_{0,0}  {\bs{u}}_{0,0}, ~ 	\frac{{\rm{d}} \alpha}{{\rm{d}} x} = 1,~~ {\mbox{with }} \bsf_{0,0} = (\bsu_{0,0} \bsu_{-1,0}+\alpha)^{-1}.$$
	We show how these symmetries can be used to reduce the equation to a six-dimensional map, a consequence of both the equation's scaling symmetry and the non-commutativity of the variables. Furthermore, we construct an auto-B\"acklund transformation via a Darboux transformation that leaves the Lax pair covariant. This transformation involves an auxiliary function, and its superposition principle yields the noncommutative Yang--Baxter map $F_{III}$ \cite{Dol}.
	
	We also show that the symmetries of both equations can be mapped, via Miura transformations (see \eqref{eq:Miura-pot} and \eqref{eq:Miura} below), to the modified Volterra equation \cite{Bog} and its master symmetry,
	\begin{equation*}
		\frac{{\rm{d}} {\bs{v}}_{0}}{{\rm{d}} t_1} =   {\bs{v}}_{0} \left( {\bs{v}}_{1} - {\bs{v}}_{-1}\right) {\bs{v}}_{0} ~~~ {\mbox{and}} ~~~ \frac{{\rm{d}} {\bs{v}}_{0}}{{\rm{d}} x} =   {\bs{v}}_{0} \left( (n+1) {\bs{v}}_{1} - (n-1) {\bs{v}}_{-1}\right) {\bs{v}}_{0},
	\end{equation*}
	respectively.
			
	The paper is organised as follows. The next section sets the framework, providing key definitions and establishing our notation. Section~\ref{sec:pKdV} deals with the symmetries and reductions of the discrete potential KdV equation, while Section~\ref{sec:KdV} discusses the symmetries and the auto-B\"acklund transformation of Hirota's KdV equation. The concluding section summarises our main results, presents additional examples (such as the noncommutative discrete nonlinear Schr\"odinger system \cite{KRX}), and outlines possible directions for future research. Finally, the appendix contains detailed derivations and proofs of various statements made in the main body of the paper.

	\section{Noncommutative setting and notation} \label{sec:notation}
	
	In this short section, we present all the necessary definitions to make our presentation self-contained, and introduce our notation. 
	
	In the following sections, we consider difference and differential--difference equations with noncommutative variables taking values in a division ring ${\mathfrak{U}}$ over a field of constants $\mathbb{F}$ of characteristic zero. The elements of ${\mathfrak{U}}$ do not necessarily commute, and the field $\mathbb{F}$ lies in the centre ${\cal{Z}}({\mathfrak{U}})$ of ${\mathfrak{U}}$. We denote elements of ${\cal{Z}}({\mathfrak{U}})$ by italic letters, and all other elements of ${\mathfrak{U}}$ by bold letters. In particular, we use the symbol $1$ to denote both the unit of the field and the unit of the ring, trusting that this will not lead to confusion.
	
	The equations we are interested in involve functions depending on the two discrete variables $n$, $m$, and their dependence on these variables is denoted with indices, i.e., $\bsu(n+i,m+j) = \bsu_{i,j}$, $i,j \in {\mathbb{Z}}$. There are two automorphisms (shift operators) $\cal{S}$ and $\cal{T}$ acting on $\bsu_{i,j}$ as ${\cal{S}}^r{\cal{T}}^s(\bsu_{i,j}) = {\cal{T}}^s{\cal{S}}^r(\bsu_{i,j}) =\bsu_{i+r,j+s}$. We also have an involution on $\mathfrak{U}$ defined in the following way.
	
	\begin{definition}
		An $\mathbb{F}$-linear map $\tau : {\mathfrak{U}} \rightarrow {\mathfrak{U}}$ satisfying $\tau(\alpha) = \alpha$, for any $\alpha \in {\mathbb{F}}$, $\tau(\bsu) = \bsu$ for any generator $\bsu$ of ${\mathfrak{U}}$, and $\tau(\bsx \bsy) = \tau(\bsy) \tau(\bsx)$, for any $\bsx, \bsy \in {\mathfrak{U}}$, is called a transposition.
	\end{definition}
	
	In order to define symmetries, we need the notion of derivation.
	
	\begin{definition}
		An $\mathbb{F}$-linear map $\partial : {\mathfrak{U}} \rightarrow {\mathfrak{U}}$ satisfying $\partial(\alpha) = 0$, for any $\alpha \in {\mathbb{F}}$, and the Leibnitz rule, i.e., $\partial(\bsx \bsy) = \partial(\bsx) \bsy +  \bsx \partial(\bsy)$ for any $\bsx, \bsy \in {\mathfrak{U}}$, is called a derivation of ${\mathfrak{U}}$.
	\end{definition}

	Derivations $\partial_\bsF$ of algebra ${\mathfrak{U}} $, commuting with the automorphisms ${\cal{S}}$ and ${\cal{T}}$, are called evolutionary. It is sufficient to define an evolutionary derivation on $\bsu_{0,0}$ as $\partial_\bsF(\bsu_{0,0}) = \bsF_{0,0} \in {\mathfrak{U}}$. The evolutionary derivation $\partial_\bsF$ is in one-to-one correspondence with the system of differential-difference equations 
	$$ \frac{{\rm{d}} u_{i,j}}{{\rm{d}} t} = {\cal{S}}^i {\cal{T}}^j \left(\bsF_{0,0}\right) = \bsF_{i,j}, ~~~ i,j \in \mathbb{Z}.$$ 
	By a symmetry of a difference equation we understand an evolutionary derivation $\partial_\bsF$ which is compatible with the equation, or, equivalently, the evolutionary differential-difference system of equations $\partial_{t} \bsu_{0,0} = \bsF_{0,0}$ is consistent with the difference equation. Symmetries $\partial_\bsF$ and $\partial_\bsG$ are called commuting, if $[\partial_\bsF,\partial_\bsG] = 0$. Symmetry $\partial_\bsM$ is called a master symmetry if $[\partial_\bsM,\partial_\bsF] \ne 0$ and $[\partial_\bsF,[\partial_\bsM,\partial_\bsF]]= 0$.

	\section{The noncommutative discrete potential KdV equation} \label{sec:pKdV}
	
	The noncommutative discrete potential KdV equation
	\begin{equation} \label{eq:nonAbel-pKdV}
		\left({\bs{u}}_{0,0} - {\bs{u}}_{1,1}\right) \left( {\bs{u}}_{1,0}-{\bs{u}}_{0,1}\right) = \alpha - \beta
	\end{equation}
	was first introduced in \cite{FNC}, where its initial value problem was also analysed. Additionally, it was derived  as a reduction of the noncommutative discrete Toda chain in \cite{Betal}.
	
	It can be readily verified that equation \eqref{eq:nonAbel-pKdV} is invariant under the following transformations, 
	\begin{equation} \label{eq:point-sym-pKdV}
		\bsu_{i,j} \rightarrow \bsu_{i,j} + \epsilon_1, ~~~ \bsu_{i,j} \rightarrow \bsu_{i,j} + \epsilon_2 (-1)^{n+m+i+j},
	\end{equation}
	as well as the conjugation $\bsu_{i,j} \rightarrow \bsx \bsu_{i,j} \bsx^{-1}$. We are interested in finding the lowest order generalised symmetries of \eqref{eq:nonAbel-pKdV}, and for this purpose we employ its Lax pair, given by
	\begin{subequations} \label{eq:Lax-pKdV}
		\begin{equation}
			{\bs{\Psi}}_{1,0} = {\bs{L}}_{0,0} {\bs{\Psi}}_{0,0} =   \begin{pmatrix}
				\bsu_{0,0} &   \alpha -\lambda -\bsu_{0,0} \bsu_{1,0} \\ 1 & -\bsu_{1,0}
			\end{pmatrix} {\bs{\Psi}}_{0,0}, \label{eq:Lax-pKdV-1}
		\end{equation}
		\begin{equation}	
			{\bs{\Psi}}_{0,1} = {\bs{M}}_{0,0} {\bs{\Psi}}_{0,0} = \begin{pmatrix}
				\bsu_{0,0} & \beta-\lambda -\bsu_{0,0} \bsu_{0,1} \\ 1 & -\bsu_{0,1}
			\end{pmatrix} {\bs{\Psi}}_{0,0}. \label{eq:Lax-pKdV-2}
		\end{equation}
	\end{subequations}
	To construct generalised symmetries of \eqref{eq:nonAbel-pKdV}, we seek an evolutionary derivation of the form
	$$\tfrac{{\rm{d}}}{{\rm{d}} t}{\bs{\Psi}}_{0,0} = {\bs{A}}_{0,0} {\bs{\Psi}}_{0,0},$$
	which is compatible with \eqref{eq:Lax-pKdV-1}. This leads to the compatibility condition
	$$ \frac{{\rm{d}}}{{\rm{d}} t} {\bs{L}}_{0,0} +  {\bs{L}}_{0,0} {\bs{A}}_{0,0} = {\bs{A}}_{1,0} {\bs{L}}_{0,0},$$
	where the derivative acts on $\bsu_{0,0}$ and its shifts, as well as on $\alpha$, all of which appear in the elements of matrix ${\bs{L}}$.
	
	Assuming that ${\bs{A}}_{0,0}  = {\bs{L}}^{-1}_{0,0} {\bs{Q}}_{0,0}$, where ${\bs{L}}^{-1}_{0,0}$ is  the inverse of ${\bs{L}}_{0,0}$ and ${\bs{Q}}_{0,0}$ is independent of $\lambda$, 
	$${\bs{L}}^{-1}_{0,0} = \frac{1}{ \alpha-\lambda} \begin{pmatrix}
		\bsu_{1,0} &  \alpha -\lambda - \bsu_{1,0} \bsu_{0,0} \\ 1 & -\bsu_{0,0}
	\end{pmatrix} \quad {\mbox{and}} \quad {\bs{Q}}_{0,0} = \begin{pmatrix}
		{\bs{a}}_{0,0} & {\bs{b}}_{0,0} \\ {\bs{c}}_{0,0} & {\bs{d}}_{0,0} \end{pmatrix}, $$
	the compatibility condition becomes
	\begin{equation} \label{eq:cc-1}
		{\bs{L}}_{1,0} \frac{{\rm{d}} {\bs{L}}_{0,0}}{{\rm{d}} t} +  {\bs{L}}_{1,0} {\bs{Q}}_{0,0} = {\bs{Q}}_{1,0} {\bs{L}}_{0,0}.
	\end{equation}
	The $\lambda$-independence of matrix $\bs{Q}$ and the $(1,1)$ entry of \eqref{eq:cc-1} yield that ${\bs{c}}_{0,0}=0$. In view of this, entries $(2,1)$ and $(1,1)$ of \eqref{eq:cc-1} lead to 
	\begin{equation} \label{eq:cc-1a}
		\frac{{\rm{d}} {\bs{u}}_{0,0}}{{\rm{d}} t} = {\bs{d}}_{1,0} - {\bs{a}}_{0,0} ~~{\text{and}} ~~ {\bs{b}}_{0,0} = \bsu_{0,0} {\bs{d}}_{0,0} - {\bs{a}}_{0,0} \bsu_{-1,0}.
	\end{equation}
	The $(1,2)$ element of \eqref{eq:cc-1} is linear in $\lambda$. Requiring the coefficient of $\lambda$ to be zero and taking into account the first equation in \eqref{eq:cc-1a} we find that ${\bs{d}}_{2,0} = {\bs{d}}_{0,0}$ which implies that ${\bs{d}}$ is periodic with period 2. Since this function corresponds to the translation symmetries \eqref{eq:point-sym-pKdV} of \eqref{eq:nonAbel-pKdV}, we choose ${\bs{d}}_{0,0} =0$ without loss of generality. The $\lambda$ independent part of the $(1,2)$ element and relations \eqref{eq:cc-1a} lead to
	$$ {\bs{a}}_{0,0} (\bsu_{1,0} - \bsu_{-1,0}) - (\bsu_{2,0}-\bsu_{0,0}) {\bs{a}}_{1,0} + \xi = 0, ~~\text{where } ~ \xi = \frac{{\rm{d}} \alpha}{{\rm{d}} t}, $$
	which imply that\footnote{The choice of the negative sign is for convenience.} 
	\begin{equation}
		{\bs{a}}_{0,0} = -(\mu_1 + \mu_2 n)  (\bsu_{1,0} - \bsu_{-1,0})^{-1} ~~ {\text{and}} ~~ \xi =  -\mu_2, ~~~ \mu_i \in {\mathbb{F}}. 
	\end{equation} 
	In view of these, equations \eqref{eq:cc-1a} become
	\begin{subequations} \label{eq:cc-1b}
		\begin{equation}\label{eq:cc-1b1}
			\frac{{\rm{d}} {\bs{u}}_{0,0}}{{\rm{d}} t} = (\mu_1 + \mu_2 n)  (\bsu_{1,0} - \bsu_{-1,0})^{-1}, ~~~ \frac{{\rm{d}} \alpha}{{\rm{d}} t} =-\mu_2,
		\end{equation} 
		and
		\begin{equation}
			{\bs{b}}_{0,0} =  (\mu_1 + \mu_2 n)  (\bsu_{1,0} - \bsu_{-1,0})^{-1}  \bsu_{-1,0},
		\end{equation}
	\end{subequations}
	respectively. This means that there are two evolutionary derivations compatible with \eqref{eq:Lax-pKdV-1} corresponding to the two parameters appearing in \eqref{eq:cc-1b}. Summarising,
	
	\begin{proposition}
		The systems 
		\begin{equation} \label{eq:Lax-difdif-pKdV-1}
			{\bs{\Psi}}_{1,0} = {\bs{L}}_{0,0} {\bs{\Psi}}_{0,0}, ~~~ \frac{{\rm{d}} {\bs{\Psi}}_{0,0}}{{\rm{d}} t_1} = {\bs{A}}_{0,0} {\bs{\Psi}}_{0,0}, ~~~ \frac{{\rm{d}} {\bs{\Psi}}_{0,0}}{{\rm{d}} x} = n {\bs{A}}_{0,0} {\bs{\Psi}}_{0,0},
		\end{equation}
		where ${\bs{L}}_{0,0}$ is given in \eqref{eq:Lax-pKdV-1} and
		\begin{equation}
			{\bs{A}}_{0,0} = \frac{1}{ \alpha-\lambda} \begin{pmatrix}
				-\bsu_{1,0} \bsa_{0,0} & \bsu_{1,0} \bsa_{0,0} \bsu_{-1,0} \\ -\bsa_{0,0} & \bsa_{0,0}\bsu_{-1,0}
			\end{pmatrix} ~~~{\mbox{with }}~~ \bsa_{0,0} = (\bsu_{1,0}-\bsu_{-1,0})^{-1},
		\end{equation}
		are Lax pairs of the differential-difference equations
		\begin{equation} \label{eq:difdif-pKdV-1}
			\frac{{\rm{d}} {\bs{u}}_{0,0}}{{\rm{d}} t_1} =  {\bs{a}}_{0,0} ~~~  {\text{ and }} ~~~  \frac{{\rm{d}} {\bs{u}}_{0,0}}{{\rm{d}} x} =  n {\bs{a}}_{0,0}, ~~ \frac{{\rm{d}} \alpha}{{\rm{d}} x} = -1,
		\end{equation}
		respectively.
	\end{proposition}
	
	Lattices \eqref{eq:difdif-pKdV-1} are related to the noncommutative modified Volterra chain via a Miura transformation. More precisely,
	
	\begin{proposition}
		The Miura transformation 	
		\begin{equation} \label{eq:Miura-pot}
			{\bs{v}}_{0,0} = \left({\bs{u}}_{1,0} -{\bs{u}}_{-1,0}\right)^{-1},
		\end{equation}
		together with the change of variables $t_1 \rightarrow -t_1$ and $x \rightarrow -x$, maps lattices \eqref{eq:difdif-pKdV-1} to
		\begin{equation} \label{eq:Volt1}
			\frac{{\rm{d}} {\bs{v}}_{0,0}}{{\rm{d}} t_1} =   {\bs{v}}_{0,0} \left( {\bs{v}}_{1,0} - {\bs{v}}_{-1,0}\right) {\bs{v}}_{0,0}
		\end{equation}
		and
		\begin{equation} \label{eq:mast-sym-mod-Volt}
			\frac{{\rm{d}} {\bs{v}}_{0,0}}{{\rm{d}} x} =  {\bs{v}}_{0,0} \left((n+1) {\bs{v}}_{1,0} - (n-1) {\bs{v}}_{-1,0}\right) {\bs{v}}_{0,0},
		\end{equation}
		respectively. Moreover, this transformation maps matrix Lax pairs \eqref{eq:Lax-difdif-pKdV-1} to the scalar Lax pair
		\begin{equation}  \label{eq:Lax-pair-Volt}
			\bs{\phi}_{2,0} + \bsv_{1,0}^{-1} \bs{\phi}_{1,0} = \lambda \bs{\phi}_{0,0}, ~~~ \frac{{\rm{d}} {\bs{\phi}}_{0,0}}{{\rm{d}} t_1} = \frac{\bsv_{0,0} \bs{\phi}_{1,0} + \bs{\phi}_{0,0}}{\lambda}
		\end{equation}
		for equation \eqref{eq:Volt1}, and to the scalar Lax pair
		\begin{equation}  \label{eq:Lax-pair-mast-Volt}
			\bs{\phi}_{2,0} + \bsv_{1,0}^{-1} \bs{\phi}_{1,0} = \lambda \bs{\phi}_{0,0},~~~ \frac{{\rm{d}} {\bs{\phi}}_{0,0}}{{\rm{d}} x} = n \frac{\bsv_{0,0} \bs{\phi}_{1,0} + \bs{\phi}_{0,0}}{\lambda}, ~~ \frac{{\rm{d}} \lambda}{{\rm{d}} x} = 1
		\end{equation}
		for equation \eqref{eq:mast-sym-mod-Volt}.
	\end{proposition}
	
	\begin{proof}
		If we rewrite the Miura transformation \eqref{eq:Miura-pot} in the form $	{\bs{v}}_{0,0}^{-1} = {\bs{u}}_{1,0} -{\bs{u}}_{-1,0}$, and differentiate it with respect to $t_1$ or $x$, while taking into account that the lattices \eqref{eq:difdif-pKdV-1} can be expressed as $\tfrac{{\rm{d}} {\bs{u}}_{0,0}}{{\rm{d}} t_1} =  {\bs{v}}_{0,0}$ and $\tfrac{{\rm{d}} {\bs{u}}_{0,0}}{{\rm{d}} x} =  n {\bs{v}}_{0,0}$, respectively, then the lattices \eqref{eq:Volt1} and \eqref{eq:mast-sym-mod-Volt} follow.	
		
		For the Lax pairs, we consider the discrete part of \eqref{eq:Lax-difdif-pKdV-1} with ${\bs{\Psi}}_{0,0} = \left( {\bs{\psi}}_{0,0} ~ {\bs{\phi}}_{0,0}\right)^T$, i.e.,
		\begin{equation} \label{eq:proof-1}
			{\bs{\psi}}_{1,0} = \bsu_{0,0} {\bs{\psi}}_{0,0} + (\alpha-\lambda - \bsu_{0,0} \bsu_{1,0}) {\bs{\phi}}_{0,0}, ~~~  
			{\bs{\phi}}_{1,0} =  {\bs{\psi}}_{0,0} -  \bsu_{1,0} {\bs{\phi}}_{0,0}.
		\end{equation} 
		The second equation readily leads to $ {\bs{\psi}}_{0,0} = {\bs{\phi}}_{1,0}+\bsu_{1,0} {\bs{\phi}}_{0,0}$, in view of which, the first equation in \eqref{eq:proof-1} becomes
		\begin{equation} \label{eq:proof-2}
			{\bs{\phi}}_{2,0}+\left(\bsu_{2,0} -\bsu_{0,0} \right) {\bs{\phi}}_{1,0}  = (\alpha-\lambda) {\bs{\phi}}_{0,0}.
		\end{equation}
		Moreover, we consider the second component of the differential equations in \eqref{eq:Lax-difdif-pKdV-1}, i.e.,
		$$\frac{{\rm{d}} {\bs{\phi}}_{0,0}}{{\rm{d}} t_1} = \frac{1}{\alpha-\lambda} \left( -\bsa_{0,0} {\bs{\psi}}_{0,0} +\bsa_{0,0} \bsu_{-1,0} {\bs{\phi}}_{0,0} \right),~~~ \frac{{\rm{d}} {\bs{\phi}}_{0,0}}{{\rm{d}} x} = \frac{n}{\alpha-\lambda} \left( -\bsa_{0,0} {\bs{\psi}}_{0,0} +\bsa_{0,0} \bsu_{-1,0} {\bs{\phi}}_{0,0} \right).$$
		Substituting $ {\bs{\psi}}_{0,0} = {\bs{\phi}}_{1,0}+\bsu_{1,0} {\bs{\phi}}_{0,0}$, the above relations become
		\begin{equation} \label{eq:proof-3}
			\frac{{\rm{d}} {\bs{\phi}}_{0,0}}{{\rm{d}} t_1} = \frac{1}{\alpha-\lambda} \left( -\bsa_{0,0} {\bs{\phi}}_{1,0} - {\bs{\phi}}_{0,0} \right),~~~ \frac{{\rm{d}} {\bs{\phi}}_{0,0}}{{\rm{d}} x} = \frac{n}{\alpha-\lambda} \left(-\bsa_{0,0} {\bs{\phi}}_{1,0} - {\bs{\phi}}_{0,0}  \right).
		\end{equation}
		Finally, we use transformation \eqref{eq:Miura-pot} along with $t_1 \rightarrow -t_1$ and $x \rightarrow -x$, as well as we change $\alpha-\lambda$ to $\lambda$. In view of all these changes, \eqref{eq:proof-2} and \eqref{eq:proof-3} lead to the Lax pairs \eqref{eq:Lax-pair-Volt} and \eqref{eq:Lax-pair-mast-Volt}, respectively. 
	\end{proof}
	
	\begin{remark}
		Lattice \eqref{eq:mast-sym-mod-Volt} is the master symmetry for the modified Volterra equation \eqref{eq:Volt1}. Indeed, their commutators yields
		\begin{equation} \label{eq:Volt2}
			\frac{{\rm{d}} {\bs{v}}_{0,0}}{{\rm{d}} t_2} =   {\bs{v}}_{0,0} {\bs{v}}_{1,0} \left( {\bs{v}}_{2,0} + {\bs{v}}_{0,0}\right) {\bs{v}}_{1,0} {\bs{v}}_{0,0} -  {\bs{v}}_{0,0} {\bs{v}}_{-1,0} \left( {\bs{v}}_{0,0} + {\bs{v}}_{-2,0}\right) {\bs{v}}_{-1,0} {\bs{v}}_{0,0},
		\end{equation}
		which is the second member of the modified Volterra hierarchy and commutes with \eqref{eq:Volt1}. It is worth noting that the Volterra lattice \eqref{eq:Volt1} was first introduced in \cite{Bog}, and corresponds to the equation denoted by ${\rm mVL}^2$ in \cite{A1} with $\alpha = 0$; see also \cite{NW}.
	\end{remark}
	
	Moreover, differential-difference equations \eqref{eq:difdif-pKdV-1} are compatible with the discrete potential KdV equation \eqref{eq:nonAbel-pKdV} and we refer to the former as generalised symmetries of the latter. 
	
	\begin{proposition}
		The lowest order generalised symmetries of equation \eqref{eq:nonAbel-pKdV} in the first direction are generated by
		\begin{equation} \label{eq:sym1-pKdV}
			\frac{{\rm{d}} {\bs{u}}_{0,0}}{{\rm{d}} t_1} =  \left({\bs{u}}_{1,0}- {\bs{u}}_{-1,0}\right)^{-1},
		\end{equation}
		and
		\begin{equation} \label{eq:mast-sym-pKdV}
			\frac{{\rm{d}} {\bs{u}}_{0,0}}{{\rm{d}} x} =  n \left({\bs{u}}_{1,0}- {\bs{u}}_{-1,0}\right)^{-1}, ~~~ 	\frac{{\rm{d}} \alpha}{{\rm{d}} x} = -1,
		\end{equation}
		respectively.
	\end{proposition}
	
	\begin{proof}
		Differentiating \eqref{eq:nonAbel-pKdV} with respect to $t_1$ and using \eqref{eq:sym1-pKdV}, we arrive at
		\begin{eqnarray}
			&& \left(\left({\bs{u}}_{1,0}- {\bs{u}}_{-1,0}\right)^{-1} - \left({\bs{u}}_{2,1}- {\bs{u}}_{0,1}\right)^{-1}\right) \left( \bsu_{1,0}- \bsu_{0,1}\right) + \nonumber \\ && \hspace{3cm} \left(\bsu_{0,0} - \bsu_{1,1}\right) \left( \left({\bs{u}}_{2,0}- {\bs{u}}_{0,0}\right)^{-1}- \left({\bs{u}}_{1,1}- {\bs{u}}_{-1,1}\right)^{-1}\right) = 0. \label{eq:pKdV-det-eq}	
		\end{eqnarray}	
		We can replace $\bsu_{2,1}$ using the forward shift of equation \eqref{eq:nonAbel-pKdV}. Indeed, the latter can be written as
		\begin{subequations} \label{eq:p-subs}
			$$ \bsu_{2,1} = \bsu_{1,0} - (\alpha-\beta) \left(\bsu_{2,0}-\bsu_{1,1}\right)^{-1}  ~~\Rightarrow ~~  \bsu_{2,1} -\bsu_{0,1}= \bsu_{1,0} -\bsu_{0,1} - (\alpha-\beta) \left(\bsu_{2,0}-\bsu_{1,1}\right)^{-1}.$$
			Using \eqref{eq:nonAbel-pKdV} to replace $\bsu_{1,0} -\bsu_{0,1}$, we arrive at
			\begin{eqnarray*}
				\bsu_{2,1} -\bsu_{0,1} &=& (\alpha-\beta)\left( \left(\bsu_{0,0} -\bsu_{1,1}\right)^{-1} -  \left(\bsu_{2,0}-\bsu_{1,1}\right)^{-1}\right) \\  &=&  (\alpha-\beta) \left(\bsu_{0,0} -\bsu_{1,1}\right)^{-1} \left(\bsu_{2,0}-\bsu_{0,0}\right) \left(\bsu_{2,0}-\bsu_{1,1}\right)^{-1}.
			\end{eqnarray*}
			Hence,
			\begin{equation} \label{eq:p-subs-a}
				\left( \bsu_{2,1} -\bsu_{0,1}\right)^{-1} =  (\alpha-\beta)^{-1}  \left(\bsu_{2,0}-\bsu_{1,1}\right) \left(\bsu_{2,0}-\bsu_{0,0}\right)^{-1} \left(\bsu_{0,0} -\bsu_{1,1}\right).
			\end{equation}
			Shifting the latter relation backward in $n$ we get that
			\begin{equation}  \label{eq:p-subs-b}
				\left( \bsu_{1,1} -\bsu_{-1,1}\right)^{-1} =  (\alpha-\beta)^{-1}  \left(\bsu_{1,0}-\bsu_{0,1}\right) \left(\bsu_{1,0}-\bsu_{-1,0}\right)^{-1} \left(\bsu_{-1,0} -\bsu_{0,1}\right).
			\end{equation}
		\end{subequations}
		Substituting \eqref{eq:p-subs} into \eqref{eq:pKdV-det-eq} and taking into account \eqref{eq:nonAbel-pKdV}, equation \eqref{eq:pKdV-det-eq} becomes an identity.
		
		To show that \eqref{eq:mast-sym-pKdV} is a symmetry of \eqref{eq:nonAbel-pKdV}, first we rewrite it as
		\begin{equation}\label{eq:mast-sym-pKdV-v2}
			\frac{{\rm{d}} {\bs{u}}_{0,0}}{{\rm{d}} x} =  n \frac{{\rm{d}} {\bs{u}}_{0,0}}{{\rm{d}} t_1}, ~~~ 	\frac{{\rm{d}} \alpha}{{\rm{d}} x} = -1.
		\end{equation}
		Then we differentiate the potential KdV equation with respect to $x$ which leads  to
		$$ n \frac{{\rm{d}} {\bs{Q}}}{{\rm{d}} t_1} - \left(\bsu_{2,1}- \bsu_{0,1}\right)^{-1} \left(\bsu_{1,0}- \bsu_{0,1}\right) + \left(\bsu_{0,0}- \bsu_{1,1}\right) \left(\bsu_{2,0}- \bsu_{0,0}\right)^{-1} +1 = 0, $$
		where ${\bs{Q}}$ denotes the left hand side of \eqref{eq:nonAbel-pKdV}. Taking into account that the first term vanishes modulo \eqref{eq:nonAbel-pKdV}, relations \eqref{eq:p-subs} and the potential KdV equation \eqref{eq:nonAbel-pKdV} turn the above relation into an identity.
	\end{proof}
	
	\begin{remark}
		Another symmetry of the potential KdV is generated by
		\begin{equation} \label{eq:sym1a-pKdV}
			\frac{{\rm{d}} {\bs{u}}_{0,0}}{{\rm{d}} \epsilon} =  n \left({\bs{u}}_{1,0}- {\bs{u}}_{-1,0}\right)^{-1} + \frac{1}{2(\alpha-\beta)} \bsu_{0,0}.
		\end{equation}
		This symmetry, however, cannot be derived using Lax pair \eqref{eq:Lax-pKdV}.
	\end{remark}

	Working in the same way or using the invariance of \eqref{eq:nonAbel-pKdV} under the interchanges $\bsu_{i,j} \leftrightarrow \bsu_{j,i}$, $\alpha \leftrightarrow \beta$, we can find symmetries in the other direction. More precisely, we can prove the following.
	\begin{proposition}
		Equation \eqref{eq:nonAbel-pKdV} admits three generalised symmetries in the second direction given by 
		\begin{equation} \label{eq:sym2-pKdV}
			\frac{{\rm{d}} {\bs{u}}_{0,0}}{{\rm{d}} s_1} =  \left({\bs{u}}_{0,1} -{\bs{u}}_{0,-1}\right)^{-1},
		\end{equation}
		\begin{equation} \label{eq:mast-sym-2-pKdV}
			\frac{{\rm{d}} {\bs{u}}_{0,0}}{{\rm{d}} y} =  m \left({\bs{u}}_{0,1} -{\bs{u}}_{0,-1}\right)^{-1},~~~ \frac{{\rm{d}} \beta}{{\rm{d}} y} = -1,
		\end{equation}
		and
		\begin{equation} \label{eq:sym2a-pKdV}
			\frac{{\rm{d}} {\bs{u}}_{0,0}}{{\rm{d}} \varepsilon} =  m \left({\bs{u}}_{0,1} -{\bs{u}}_{0,-1}\right)^{-1} - \frac{1}{2 \left(\alpha-\beta\right)} \bsu_{0,0},
		\end{equation}
		respectively.
	\end{proposition}
	
	In what follows, we employ these symmetries to reduce equation \eqref{eq:nonAbel-pKdV} to a noncommutative discrete Painlevé equation and to relate some of its solutions to a system of partial differential equations..

	\subsection{Symmetry reduction and noncommutative maps} \label{sec:red-KdV}

	Equation \eqref{eq:nonAbel-pKdV} admits the five-point generalised symmetry
	$$ \bsF_{0,0} = (n + \mu_1)  \left({\bs{u}}_{1,0} -{\bs{u}}_{-1,0}\right)^{-1} + (m + \mu_2) \left({\bs{u}}_{0,1} -{\bs{u}}_{0,-1}\right)^{-1} + \lambda_1 + \lambda_2 (-1)^{n+m}.$$
	To find solutions of \eqref{eq:nonAbel-pKdV} satisfying the constraint $\bsF_{0,0} = 0$, we introduce the functions
	\begin{equation} \label{eq:invariants-add}
		{\bs{x}}_{0,0} = \bsu_{1,0} - \bsu_{0,0} , ~~  {\bs{y}}_{0,0} = \bsu_{0,1} - \bsu_{0,0},
	\end{equation}
	which satisfy the compatibility condition
	\begin{equation}\label{eq:conne-inv-add}
		{\bs{x}}_{0,1} +  {\bs{y}}_{0,0} = {\bs{y}}_{1,0} + {\bs{x}}_{0,0}.
	\end{equation}
	Using \eqref{eq:invariants-add}, their shifts, and the relation \eqref{eq:conne-inv-add}, we can rewrite equation \eqref{eq:nonAbel-pKdV} as a system for the functions $\bsx$ and $\bsy$ in either of the following equivalent forms,
	\begin{subequations} \label{eq:inv-sys-add}
		\begin{equation}\label{eq:inv-sys-A-add}
			{\bs{x}}_{0,1} = - {\bs{y}}_{0,0} - r \left({\bs{x}}_{0,0}- \bsy_{0,0}\right)^{-1}, ~~~ {\bs{y}}_{1,0} = - {\bs{x}}_{0,0} - r \left({\bs{x}}_{0,0}- \bsy_{0,0}\right)^{-1},
		\end{equation}
		or
		\begin{equation}\label{eq:inv-sys-B-add}
			{\bs{x}}_{0,0} = {\bs{y}}_{0,0} - r \left({\bs{x}}_{0,1}+ \bsy_{0,0}\right)^{-1}, ~~~ {\bs{y}}_{1,0} ={\bs{x}}_{0,1} + r \left({\bs{x}}_{0,1}+ \bsy_{0,0}\right)^{-1},
		\end{equation}
		where $r := \alpha-\beta$. In the same fashion, the symmetry constraint $\bsF_{0,0} = 0$ becomes
		\begin{equation} \label{eq:inv-sys-C-add}
			(n + \mu_1) (\bsx_{0,0}+ \bsx_{-1,0})^{-1} + (m+\mu_2) (\bsy_{0,0} + \bsy_{0,-1})^{-1} + \lambda_1 + \lambda_2 (-1)^{n+m} =0.
		\end{equation}
	\end{subequations}
	
	\begin{center}
		\begin{figure}[th]
			\centertexdraw{ \setunitscale 0.85
				\linewd 0.02 \arrowheadtype t:F 
				\htext(0 0.5) {\phantom{T}}
				\move (-.5 0) \lvec(3.5 0)
				\move (-.5 -1) \lvec(3.5 -1)
				\move (0 -1.5) \lvec (0 1.5) 
				\move (1 -1.5) \lvec (1 1.5)
				\move (-.5 1) \lvec (3.5 1)
				\move (2 -1.5) \lvec (2 1.5)
				\move (3 -1.5) \lvec (3 1.5)
				\move (1 -1) \fcir f:0.0 r:0.07 \move (2 -1) \fcir f:0.0 r:0.07
				\move (1 0) \fcir f:0.0 r:0.07 \move (0 0) \fcir f:0.0 r:0.07 \move (2 0) \fcir f:0.0 r:0.07
				\move (1 1) \fcir f:0.0 r:0.07 \move (3 0) \fcir f:0.0 r:0.07 \move (2 1) \fcir f:0.0 r:0.07
				\htext (-.5 -.27) {{\color[rgb]{.7,0,.1}{\footnotesize{$\bsu_{-1,0}$}}}} \htext (1.04 -.27) {{\color[rgb]{.7,0,.1}{\footnotesize{$\bsu_{0,0}$}}}}
				\htext (1.04 -1.27) {{\color[rgb]{.7,0,.1}{\footnotesize{$\bsu_{0,-1}$}}}}
				\htext (2.04 -.27) {{\color[rgb]{.7,0,.1}{\footnotesize{$\bsu_{1,0}$}}}}  
				\htext (2.04 -1.27) {{\color[rgb]{.7,0,.1}{\footnotesize{$\bsu_{1,-1}$}}}}
				\htext (3.04 -.27) {{\color[rgb]{.7,0,.1}{\footnotesize{$\bsu_{2,0}$}}}}
				\htext (1.05 1.1) {{\color[rgb]{.7,0,.1}{\footnotesize{$\bsu_{0,1}$}}}}		\htext (2.05 1.1) {{\color[rgb]{.7,0,.1}{\footnotesize{$\bsu_{1,1}$}}}}	
				\htext (.4 -.2) {{\color[rgb]{.1,0,.7}{\footnotesize{$\bsx_{n}$}}}} \htext (1.04 .4) {{\color[rgb]{.1,0,.7}{\footnotesize{$\bsy_{n}$}}}}
				\htext (2.04 .4) {{\color[rgb]{.1,0,.7}{\footnotesize{$\bsy_{n+1}$}}}}
				\htext (1.4 -.2) {{\color[rgb]{.1,0,.7}{\footnotesize{$\bsz_{n}$}}}}  \htext (2.4 -.2) {{\color[rgb]{.1,0,.7}{\footnotesize{$\bsz_{n+1}$}}}}
				\htext (1.14 -.5) {{\color[rgb]{.1,0.7,0}{\footnotesize{$\bsx_{n+1} = \bsz_{n}$}}}}
			}
			\caption{Variables $\bsx_n = \bsu_{0,0}-\bsu_{-1,0}$, $\bsy_n= \bsu_{0,1}-\bsu_{0,0}$, $\bsz_n = \bsu_{1,0}-\bsu_{0,0}$ on the $\mathbb{Z}^2$ lattice} \label{fig:lattice}
		\end{figure}
	\end{center}
	
If the values of $\bsx_{-1,0}$, $\bsy_{0,0}$, and $\bsx_{0,0}$ are known, then the corresponding shifts in the first direction can be determined. Defining $\bsx_n := \bsx_{-1,0}$, $\bsy_n := \bsy_{0,0}$, and $\bsz_n := \bsx_{0,0}$ (see Figure~\ref{fig:lattice}), we can compute the updated variables $\bsx_{n+1}$, $\bsy_{n+1}$, and $\bsz_{n+1}$ by using the relations \eqref{eq:inv-sys-add} and their shifts.

In particular, from the identification $\bsx_{n+1} = \bsz_n$ and the second equation in \eqref{eq:inv-sys-A-add}, we obtain
	\begin{equation} \label{eq:inv-sys-D-add}
		{\bs{x}}_{n+1} = {\bs{z}}_n, ~~~ {\bs{y}}_{n+1} = -{\bs{z}}_n -r ({\bs{z}}_{n}-{\bs{y}}_{n})^{-1}.
	\end{equation}
	To determine $\bsz_{n+1}$, we first rewrite the symmetry constraint \eqref{eq:inv-sys-C-add} as
	\begin{equation} \label{eq:inv-sys-D1-add}
		(n + \mu_1) (\bsz_n+ \bsx_n)^{-1} + (m+\mu_2) (\bsy_n + \bsy_{0,-1})^{-1} + \lambda_1 + \lambda_2 (-1)^{n+m} =0,
	\end{equation}
	and then consider its forward shift in $n$, 
	\begin{equation} \label{eq:inv-sys-D2-add}
		(n +1 + \mu_1) (\bsz_{n+1}+ \bsz_n)^{-1} + (m+\mu_2) (\bsy_{n+1} + \bsy_{1,-1})^{-1} + \lambda_1 - \lambda_2 (-1)^{n+m} =0.
	\end{equation}
	The term $\bsy_{n+1} + \bsy_{1,-1}$ can be expressed using the backward shift in $m$ of the second equation in \eqref{eq:inv-sys-B-add}, i.e. 	${\bs{y}}_{1,-1} ={\bs{z}}_n + r \left({\bs{z}}_{n}+ \bsy_{0,-1}\right)^{-1}$, combined with the second equation of \eqref{eq:inv-sys-D-add}. This yields the identity 
	$$ \left(\bsy_{n+1} + \bsy_{1,-1}\right)^{-1} = r^{-1}  (\bsz_n+\bsy_{0,-1}) (\bsy_n + \bsy_{0,-1})^{-1} (\bsy_n-\bsz_n). $$
	Substituting into \eqref{eq:inv-sys-D2-add}, and using \eqref{eq:inv-sys-D1-add} to eliminate $\bsy_{0,-1}$, we derive
	\begin{equation} \label{eq:inv-sys-D3-add}
		r (n+1+\mu_1) ({\bs{z}}_{n+1}+{\bs{z}}_n)^{-1} + ({\bs{z}}_n-\bsy_n) {\bs{A}}  ({\bs{z}}_n-\bsy_n) - (m+\mu_2)  ({\bs{z}}_n-\bsy_n) + r (\lambda_1 - \lambda_2 (-1)^{n+m}) =0,
	\end{equation}
	where ${\bs{A}} = (n+ \mu_1) ({\bs{z}}_n + \bsx_n)^{-1} + \lambda_1 + \lambda_2 (-1)^{n+m}$. Rearranging this equation for $\bsz_{n+1}$, together with \eqref{eq:inv-sys-D-add}, defines a three-dimensional map for the variables $\bsx_n$, $\bsy_n$, and $\bsz_n$.
	
	This map is related to a discrete Painlev\`e-type equation for the function $\bsY_n := \bsy_{n+1} + \bsz_n$. In particular, the quantity $\bsz_n - \bsy_n$ can be expressed in terms of $\bsY_n$ via the second equation in \eqref{eq:inv-sys-D-add}. Moreover, from \eqref{eq:inv-sys-D-add} we obtain
	$$(\bsz_{n+1}+\bsz_n)^{-1} = \left(\bsY_{n+1} \bsY_n -r \right)^{-1} \bsY_{n+1}, ~~~ (\bsz_{n}+\bsx_{n})^{-1} = \bsY_n \left(\bsY_{n-1} \bsY_n -r \right)^{-1}.$$
	Substituting these expressions into \eqref{eq:inv-sys-D3-add} and simplifying yields the following noncommutative difference equation
	\begin{equation} \label{eq:dPII}
		r f_{n+1} \left( \bsY_{n+1} \bsY_n - r \right)^{-1} + r f_{n} \left( \bsY_{n-1} \bsY_n - r \right)^{-1}  + r \lambda_{n+m} \bsY_n^{-1} + \lambda_{n+m+1} \bsY_n + g_m + f_{n+1} = 0,
	\end{equation}
	where $f_n = n + \mu_1$, $g_m = m + \mu_2$, and $\lambda_k = \lambda_1 + \lambda_2 (-1)^k$.  We refer to this equation as the noncommutative asymmetric alternate discrete Painlev\`e II equation, as it reduces to the asymmetric alternate discrete Painlev\`e II equation \cite{FGR} when the variable $\bsY_n$ is assumed to be a commuting variable.

	\subsection{Continuous symmetric reductions and a system of differential equations} \label{sec:rpde}
	
	Solutions of the discrete potential KdV equation generally depend on the parameters $\alpha$ and $\beta$. In this section, we focus on a special class of solutions of \eqref{eq:nonAbel-pKdV} that depends on these parameters and remain invariant under both master symmetries given in \eqref{eq:mast-sym-pKdV} and \eqref{eq:mast-sym-2-pKdV}. We refer to such solutions as continuously symmetric solutions, as they satisfy the following symmetry constraints in addition to equation \eqref{eq:nonAbel-pKdV}.
	\begin{equation}\label{eq:con-const}
		\frac{\partial {\bs{u}}_{0,0}}{\partial \alpha} =  -n \left({\bs{u}}_{1,0} -{\bs{u}}_{-1,0}\right)^{-1}, ~~~ 	\frac{\partial {\bs{u}}_{0,0}}{\partial \beta} =  -m \left({\bs{u}}_{0,1} -{\bs{u}}_{0,-1}\right)^{-1}
	\end{equation}
	Since this system involves six shifts of the field $\bsu$, we may eliminate three of them to obtain a closed system for the remaining variables. In what follows, we choose to eliminate $\bsu_{1,1}$, $\bsu_{-1,0}$, and $\bsu_{0,-1}$, thereby reducing the system to one involving only $\bsu_{0,0}$, $\bsu_{1,0}$, and $\bsu_{0,1}$. The elimination can be carried out systematically and the full derivation is given in the Appendix.
	
	To simplify notation, we denote the remaining variables by $\bsu := \bsu_{0,0}$, $\bsu_1 := \bsu_{1,0}$, and $\bsu_2 := \bsu_{0,1}$. In terms of these variables, we now state the following proposition.
	\begin{proposition}
		System 
		\begin{subequations} \label{eq:rpde}
			\begin{eqnarray}
				\frac{\partial {\bs{u}}_{1}}{\partial \beta} &=& \frac{1}{\alpha-\beta} \left({\bs{u}}_{1} -{\bs{u}}_{2}\right) \left(m - \frac{\partial {\bs{u}}}{\partial \beta}  \left( {\bs{u}}_{1} -{\bs{u}}_{2} \right)\right), \label{eq:rpde-1} \\
				\frac{\partial {\bs{u}}_{2}}{\partial \alpha} &=&\frac{1}{\alpha-\beta} \left({\bs{u}}_{1} -{\bs{u}}_{2}\right) \left(n+ \frac{\partial {\bs{u}}}{\partial \alpha}  \left( {\bs{u}}_{1} -{\bs{u}}_{2} \right)\right),  \label{eq:rpde-2}\\
				\frac{\partial^2 {\bs{u}}}{\partial \alpha \partial \beta} &=& \frac{1}{\alpha-\beta} \left(  \frac{\partial {\bs{u}}}{\partial \alpha} \left({\bs{u}}_{1} -{\bs{u}}_{2}\right)  \frac{\partial {\bs{u}}}{\partial \beta} + \frac{\partial {\bs{u}}}{\partial \beta} \left({\bs{u}}_{1} -{\bs{u}}_{2}\right)  \frac{\partial {\bs{u}}}{\partial \alpha} + n \frac{\partial {\bs{u}}}{\partial \beta} -  m \frac{\partial {\bs{u}}}{\partial \alpha} \right), \label{eq:rpde-3}
			\end{eqnarray}
		\end{subequations}
		determines the continuously symmetric solutions of \eqref{eq:nonAbel-pKdV}.
	\end{proposition}
	
	This is an integrable system in the sense  that it admits a Lax pair and an auto-B\"acklund transformation.
	
	\begin{proposition}
		A Lax pair for system \eqref{eq:rpde} is given by
		\begin{equation} \label{eq:Lax-rpde}
			\frac{\partial {\bs{\Psi}}}{\partial \alpha} = \frac{1}{\alpha-\lambda} \begin{pmatrix}
				-\bsu_1 \frac{\partial {\bs{u}}}{\partial \alpha} &  \bsu_1 \frac{\partial {\bs{u}}}{\partial \alpha} \bsu_1 + n \bsu_1 \\
				-\frac{\partial {\bs{u}}}{\partial \alpha} &  \frac{\partial {\bs{u}}}{\partial \alpha} \bsu_1 + n
			\end{pmatrix} {\bs{\Psi}}, ~~~ \frac{\partial {\bs{\Psi}}}{\partial \beta} = \frac{1}{\beta-\lambda} \begin{pmatrix}
				-\bsu_2 \frac{\partial {\bs{u}}}{\partial \beta} &  \bsu_2 \frac{\partial {\bs{u}}}{\partial \beta} \bsu_2 + m \bsu_2 \\
				-\frac{\partial {\bs{u}}}{\partial \beta} &  \frac{\partial {\bs{u}}}{\partial \beta} \bsu_2 + m
			\end{pmatrix} {\bs{\Psi}}.
		\end{equation}
		Moreover, system
		\begin{subequations} \label{eq:BT-rpde}
			\begin{eqnarray}
				&& \frac{\partial \tilde{\bsu}}{\partial \alpha} = \frac{1}{\alpha - \lambda} ({\bsu}_1- \tilde{\bsu}) \left(n + \frac{\partial \bsu}{\partial \alpha} ({\bsu}_1- \tilde{\bsu})\right)	\label{eq:BT-rpde-1} \\
				&& \frac{\partial \tilde{\bsu}}{\partial \beta} = \frac{1}{\beta - \lambda} ({\bsu}_2- \tilde{\bsu}) \left(m + \frac{\partial \bsu}{\partial \beta} ({\bsu}_2- \tilde{\bsu})\right)	\label{eq:BT-rpde-2} \\
				&& \left({\bs{u}} - \tilde{\bs{u}}_{1}\right) \left( {\bs{u}}_{1}-\tilde{\bs{u}}\right) = \alpha - \lambda,  \label{eq:BT-rpde-3} \\
				&& 	\left({\bs{u}} - \tilde{\bs{u}}_{2}\right) \left( {\bs{u}}_{2}-\tilde{\bs{u}}\right) = \beta - \lambda, 	\label{eq:BT-rpde-4}
			\end{eqnarray}
		\end{subequations}
		is an auto-B\"acklund transformation of system \eqref{eq:rpde} and is invariant under the interchange of variables $ (\bsu,\bsu_1,\bsu_2)  \leftrightarrow (\tilde{\bsu},\tilde{\bsu}_1,\tilde{\bsu}_2)$.
	\end{proposition}
	
	\begin{proof}
		The first equation of the Lax pair follows from \eqref{eq:Lax-difdif-pKdV-1} by employing \eqref{eq:con-const} to replace the negative shifts of $\bsu$. The second equation follows from the first equation and the invariance of \eqref{eq:rpde} under the interchanges $\bsu_1 \leftrightarrow \bsu_2$, $\alpha \leftrightarrow \beta$, and $n \leftrightarrow m$. The compatibility condition of \eqref{eq:Lax-rpde} splits into two equations, namely 
		$$  \partial_\beta {\bs{A}} = \partial_\alpha {\bs{B}}, ~~~ \beta  \partial_\beta {\bs{A}}  - \alpha \partial_\alpha {\bs{B}} + [{\bs{A}},{\bs{B}}] = 0,$$
		where $\bs{A}$ and $\bs{B}$ are the matrices of the first and second equation, respectively, in \eqref{eq:Lax-rpde}. The $(2,1)$ entry of the second equation yields \eqref{eq:rpde-3}, in view of which equation $  \partial_\beta {\bs{A}} = \partial_\alpha {\bs{B}}$ leads to the other two equation of \eqref{eq:rpde}. Finally,  in view of \eqref{eq:rpde}, the remaining equations of the compatibility condition hold identically.
		
		Regarding the auto-B\"acklund transformation, we employ the corresponding transformation of the discrete potential KdV equation,
		\begin{equation} \label{eq:BT-dpKdV}
			\left({\bs{u}}_{0,0} - \tilde{\bs{u}}_{1,0}\right) \left( {\bs{u}}_{1,0}-\tilde{\bs{u}}_{0,0}\right) = \alpha - \lambda, ~~~ 	\left({\bs{u}}_{0,0} - \tilde{\bs{u}}_{1,0}\right) \left( {\bs{u}}_{0,1}-\tilde{\bs{u}}_{0,0}\right) = \beta - \lambda,
		\end{equation}
		which follows from the multidimensionally consistency of \eqref{eq:nonAbel-pKdV}. To do this, let us first consider the equations following from \eqref{eq:con-const} after replacing $\bsu$ with $\tilde{\bsu}$, i.e.,
		\begin{equation} \label{eq:tilde-con-const}
		\frac{\partial \tilde{\bs{u}}_{0,0}}{\partial \alpha} =  -n \left(\tilde{\bs{u}}_{1,0} -\tilde{\bs{u}}_{-1,0}\right)^{-1}, ~~~ 	\frac{\partial \tilde{\bs{u}}_{0,0}}{\partial \beta} =  -m \left(\tilde{\bs{u}}_{0,1} - \tilde{\bs{u}}_{0,-1}\right)^{-1}.
		\end{equation}
		Clearly system \eqref{eq:BT-dpKdV} is consistent with \eqref{eq:con-const} and \eqref{eq:tilde-con-const}, as this is another manifestation of the invariance of the discrete potential KdV under symmetries \eqref{eq:mast-sym-pKdV} and \eqref{eq:mast-sym-2-pKdV}. Moreover, it follows from \eqref{eq:BT-dpKdV} and their backward shifts that
		$$ \bsu_{1,0}-\bsu_{-1,0} = \frac{1}{\alpha - \lambda} (\bsu_{0,0} - \tilde{\bsu}_{1,0})^{-1} (\tilde{\bsu}_{1,0}-\tilde{\bsu}_{-1,0}) (\bsu_{0,0} -\tilde{\bsu}_{-1,0})^{-1},$$
		and
		$$\bsu_{0,1}-\bsu_{0,-1} = \frac{1}{\beta - \lambda} (\bsu_{0,0} - \tilde{\bsu}_{0,1})^{-1} (\tilde{\bsu}_{0,1}-\tilde{\bsu}_{0,-1}) (\bsu_{0,0} -\tilde{\bsu}_{0,-1})^{-1}.$$
		In view of these relations, we can rewrite \eqref{eq:con-const} as 
		$$ \frac{\partial \bsu_{0,0}}{\partial \alpha} = \frac{1}{\alpha - \lambda} (\bsu_{0,0} -\tilde{\bsu}_{-1,0}) \frac{\partial \tilde{\bsu}_{0,0}}{\partial \alpha} (\bsu_{0,0} - \tilde{\bsu}_{1,0}),$$
		and
		$$\frac{\partial \bsu_{0,0}}{\partial \beta} = \frac{1}{\beta - \lambda} (\bsu_{0,0} -\tilde{\bsu}_{0,-1}) \frac{\partial \tilde{\bsu}_{0,0}}{\partial \beta} (\bsu_{0,0} - \tilde{\bsu}_{0,1}),$$
		respectively. Finally, we eliminate the backward shifts of $\tilde{\bsu}$ using \eqref{eq:tilde-con-const} and rename
		$$ (\bsu_{0,0}, \bsu_{1,0}, \bsu_{0,1}) \longrightarrow (\bsu,\bsu_1,\bsu_2) ~~ {\mbox{and}} ~~  (\tilde{\bsu}_{0,0}, \tilde{\bsu}_{1,0}, \tilde{\bsu}_{0,1}) \longrightarrow (\tilde{\bsu},\tilde{\bsu}_1,\tilde{\bsu}_2).$$
		These considerations lead to equations \eqref{eq:BT-rpde-1} and \eqref{eq:BT-rpde-2}. The compatibility condition of these equations yields a quadratic polynomial in $\tilde{\bsu}$ the coefficients of which yield system \eqref{eq:rpde}.
	\end{proof}

	\begin{remark}
		System \eqref{eq:rpde} constitutes a noncommutative generalisation of the system originally introduced in \cite{NHJ}, further analysed in \cite{TTX1} and \cite{TTX}, and was derived in the context of ABS equations and their continuous symmetric reductions in \cite{TX}. From system \eqref{eq:rpde}, three distinct subsystems can be extracted, each generalising a classical system from Mathematical Physics. The first one is the Euler--Poisson--Darboux (EPD) equation,
		$$\frac{\partial^2 {\bs{u}}}{\partial \alpha \partial \beta} = \frac{1}{\alpha-\beta} \left(n \frac{\partial {\bs{u}}}{\partial \beta} -  m \frac{\partial {\bs{u}}}{\partial \alpha} \right),$$
		corresponding to the choices $ \bsu_1 = \bsu_2 = c \in {\mathbb{F}}$. The other system is the Ernst equation along with the Neugebauer--Kramer involution. This corresponds to the choice $ \bsu_1 = {\bs{\phi}} + {\rm{i}} {\bs{\chi}}$, $\bsu_2 = -{\bs{\phi}} + {\rm{i}} {\bs{\chi}}$, $\bsu = \bsF + {\rm{i}} {\bs{\omega}}$, and $n=m=-\tfrac{1}{2}$. In view of these choices, equation \eqref{eq:rpde-3} becomes the Ernst equation,
		\begin{equation} \label{eq:Ernst}
			\frac{\partial^2 {\bs{u}}}{\partial \alpha \partial \beta} =  \frac{1}{2} \left(\frac{\partial {\bs{u}}}{\partial \alpha} \bsF^{-1}  \frac{\partial {\bs{u}}}{\partial \beta} + \frac{\partial {\bs{u}}}{\partial \beta} \bsF^{-1}  \frac{\partial {\bs{u}}}{\partial \alpha}\right) + \frac{1}{2 \left(\alpha-\beta\right)}  \left(\frac{\partial {\bs{u}}}{\partial \alpha} - \frac{\partial {\bs{u}}}{\partial \beta}\right),
		\end{equation}
		and the other two equations of \eqref{eq:rpde} lead to the Neugebauer--Kramer involution,
		\begin{equation} \label{eq:NK}
			\bsF {\bs{\phi}} = \frac{\alpha-\beta}{4}, ~~~ \frac{\partial {\bs{\chi}}}{\partial \alpha} = \frac{4}{\alpha-\beta} {\bs{\phi}} \frac{\partial {\bs{\omega}}}{\partial \alpha}  {\bs{\phi}},~~~ \frac{\partial {\bs{\chi}}}{\partial \beta} = \frac{-4}{\alpha-\beta} {\bs{\phi}} \frac{\partial {\bs{\omega}}}{\partial \beta}  {\bs{\phi}}.
		\end{equation}
		It can be easily checked with direct computations that if $\bsu = \bsF + {\rm{i}} {\bs{\omega}}$ satisfies \eqref{eq:Ernst} and $\bsu_1= {\bs{\phi}} + {\rm{i}} {\bs{\chi}}$ is related to $\bsu$ via \eqref{eq:NK}, then $\bsu_1$ satisfies
		$$\frac{\partial^2 {\bs{u}}_1}{\partial \alpha \partial \beta} =  \frac{1}{2} \left(\frac{\partial {\bs{u}}_1}{\partial \alpha} {\bs{\phi}}^{-1}  \frac{\partial {\bs{u}}_1}{\partial \beta} + \frac{\partial {\bs{u}}_1}{\partial \beta} {\bs{\phi}}^{-1}  \frac{\partial {\bs{u}}_1}{\partial \alpha}\right) + \frac{1}{2 \left(\alpha-\beta\right)}  \left(\frac{\partial {\bs{u}}_1}{\partial \alpha} - \frac{\partial {\bs{u}}_1}{\partial \beta}\right), $$
		and vice versa. It is worth noting that equation \eqref{eq:Ernst} was proposed as a matrix version of the Ernst equation in \cite{AL}.
		
		Finally, system \eqref{eq:rpde} can also be decoupled to a system for $\bsu_1$ and $\bsu_2$. This can be done by rearranging equations \eqref{eq:rpde-1} and \eqref{eq:rpde-2} for the first order derivatives of $\bsu$ and then consider their compatibility conditions with \eqref{eq:rpde-3}. This leads to 
		\begin{eqnarray*}
			\frac{\partial^2 {\bs{u}}_1}{\partial \alpha \partial \beta} &=& \frac{\partial {\bs{u}}_1}{\partial \alpha}\left({\bs{u}}_{1} -{\bs{u}}_{2}\right) ^{-1}   \frac{\partial {\bs{u}}_1}{\partial \beta} +  \frac{\partial {\bs{u}}_1}{\partial \beta}\left({\bs{u}}_{1} -{\bs{u}}_{2}\right) ^{-1}   \frac{\partial {\bs{u}}_1}{\partial \alpha} - \frac{m}{\alpha-\beta}  \frac{\partial {\bs{u}}_1}{\partial \alpha} - \frac{n+1}{\alpha-\beta}   \frac{\partial {\bs{u}}_1}{\partial \beta}, \\
			\frac{\partial^2 {\bs{u}}_2}{\partial \alpha \partial \beta} &=& \frac{\partial {\bs{u}}_2}{\partial \alpha}\left({\bs{u}}_{2} -{\bs{u}}_{1}\right) ^{-1}   \frac{\partial {\bs{u}}_2}{\partial \beta}  + \frac{\partial {\bs{u}}_2}{\partial \beta}\left({\bs{u}}_{2} -{\bs{u}}_{1}\right) ^{-1}   \frac{\partial {\bs{u}}_2}{\partial \alpha} + \frac{m+1}{\alpha-\beta}  \frac{\partial {\bs{u}}_2}{\partial \alpha} + \frac{n}{\alpha-\beta}   \frac{\partial {\bs{u}}_2}{\partial \beta},
		\end{eqnarray*} 
		which may be viewed as a noncommutative generalisation of the stationary Loewner--Konopelchenko--Rogers system, see \cite{TTX, WS} and references therein.
	\end{remark}

	\section{The noncommutative discrete KdV equation} \label{sec:KdV}
	
	In this section, we consider the noncommutative discrete Hirota's KdV equation, or simply the KdV equation,	
	\begin{equation} \label{eq:nonAbel-KdV}
		{\bs{u}}_{0,0} + \alpha {\bs{u}}_{1,0}^{-1} - \alpha {\bs{u}}_{0,1}^{-1}-{\bs{u}}_{1,1} = {\bs{0}}, 
	\end{equation}
	and derive its lowest order generalised symmetries and an auto-B\"acklund transformation using its Lax pair. Equation \eqref{eq:nonAbel-KdV} was derived as a reduction of the noncommutative Toda chain in \cite{Betal} and its Lax pair is
	\begin{equation} \label{eq:Lax-KdV}
		{\bs{\Psi}}_{1,0} = {\bs{L}}_{0,0} {\bs{\Psi}}_{0,0} =   \begin{pmatrix}
			\alpha \bsu_{0,0}^{-1} & \lambda  \\ \lambda &  \bsu_{0,0}
		\end{pmatrix} {\bs{\Psi}}_{0,0}, ~~~ 
		{\bs{\Psi}}_{0,1} = {\bs{M}}_{0,0} {\bs{\Psi}}_{0,0} = \begin{pmatrix}
			\alpha \bsu_{0,0}^{-1} - \bsu_{0,1} & \lambda  \\ \lambda & 0
		\end{pmatrix} {\bs{\Psi}}_{0,0}.
	\end{equation}
	It can be easily verified that \eqref{eq:nonAbel-KdV} is invariant under the scaling generated by 	$\tfrac{\rm{d}}{{\rm{d}} t} \bsu_{0,0} = (-1)^{n+m} \bsu_{0,0}$, and the conjugation $\bsu_{i,j} \rightarrow \bsx \bsu_{i,j} \bsx^{-1}$.
	
	Starting from the KdV equation and its Lax pair, and proceeding analogously to the previous section, we can determine the generalised symmetries of \eqref{eq:nonAbel-KdV}. This analysis can be summarised in the following statements.
	\begin{proposition}
		The systems 
		\begin{equation} \label{eq:Lax-difdif-KdV-1}
			{\bs{\Psi}}_{1,0} = {\bs{L}}_{0,0} {\bs{\Psi}}_{0,0}, ~~~ \frac{{\rm{d}} {\bs{\Psi}}_{0,0}}{{\rm{d}} t_1} = {\bs{A}}_{0,0} {\bs{\Psi}}_{0,0}, ~~~ \frac{{\rm{d}} {\bs{\Psi}}_{0,0}}{{\rm{d}} x} = n {\bs{A}}_{0,0} {\bs{\Psi}}_{0,0},
		\end{equation}
		where ${\bs{L}}_{0,0}$ is given in \eqref{eq:Lax-KdV} and
		\begin{equation} \label{eq:Lax-difdif-KdV-2}
			{\bs{A}}_{0,0} = \frac{1}{\alpha-\lambda^2} \begin{pmatrix}
				\bsf_{0,0} \bsu_{0,0} \bsu_{-1,0} & -\lambda \bsf_{0,0} \bsu_{0,0} \\ 
				-\lambda \bsu_{-1,0} \bsf_{0,0} & \alpha \bsu_{0,0}^{-1} \bsf_{0,0} \bsu_{0,0}
			\end{pmatrix} ~~ {\text{with}} ~~\bsf_{0,0} = (\bsu_{0,0} \bsu_{-1,0}+\alpha)^{-1},
		\end{equation}
		are Lax pairs of the differential-difference equations
		\begin{equation} \label{eq:difdif-KdV-1}
			\frac{{\rm{d}} {\bs{u}}_{0,0}}{{\rm{d}} t_1} =  {\bs{u}}_{0,0} {\bs{f}}_{1,0} -{\bs{f}}_{0,0}  {\bs{u}}_{0,0} ~~~  {\text{ and }} ~~~ \frac{{\rm{d}} {\bs{u}}_{0,0}}{{\rm{d}} x} =  n {\bs{u}}_{0,0} {\bs{f}}_{1,0} -(n-1) {\bs{f}}_{0,0}  {\bs{u}}_{0,0}, ~~~ 	\frac{{\rm{d}} \alpha}{{\rm{d}} x} = 1,
		\end{equation}
		respectively.
	\end{proposition}
	
	\begin{proof}
		It can be verified by direct calculations.
	\end{proof}
	
	\begin{proposition} \label{prop:KdV-sym-1}
		The lowest order generalised symmetries of equation \eqref{eq:nonAbel-KdV} in the first direction are generated by differential-difference equations \eqref{eq:difdif-KdV-1}.
	\end{proposition}
	
	\begin{proof}
		It is given in the Appendix.
	\end{proof}
	
	The lattices in \eqref{eq:difdif-KdV-1} are related to the Volterra lattice and its master symmetry via a Miura transformation. 
	
	\begin{proposition} \label{prop:Miura-KdV}
		The Miura transformation 	
		\begin{equation} \label{eq:Miura}
			{\bs{v}}_{0,0} = {\bs{u}}_{0,0} ({\bs{u}}_{1,0} {\bs{u}}_{0,0}+\alpha)^{-1},
		\end{equation}
		together with the change of variables $t_1 \rightarrow -t_1$ and $x \rightarrow -x$, maps lattices \eqref{eq:difdif-KdV-1} to \eqref{eq:Volt1} and \eqref{eq:mast-sym-mod-Volt}, respectively.
	\end{proposition}
	
	\begin{proof}
		It is given in the Appendix.
	\end{proof}
	
	Employing the invariance of \eqref{eq:nonAbel-KdV} under the interchange of shifts, i.e., $\bsu_{1,0} \leftrightarrow \bsu_{0,1}$, and the change of parameter $\alpha$ to $-\alpha$, we state the following result which can also be proven by straightforward calculations.
	\begin{proposition}  \label{prop:KdV-sym-2}
		The lowest order generalised symmetries of equation \eqref{eq:nonAbel-KdV} in the second direction are given by 
		\begin{equation} \label{eq:sym1a-KdV}
			\frac{{\rm{d}} {\bs{u}}_{0,0}}{{\rm{d}} s_1} =  {\bs{u}}_{0,0} {\bs{g}}_{0,1} -{\bs{g}}_{0,0}  {\bs{u}}_{0,0}, ~~~ {\bs{g}}_{0,0} = ({\bs{u}}_{0,0} {\bs{u}}_{0,-1}-\alpha)^{-1}
		\end{equation}
		and
		\begin{equation} \label{eq:mast-sym-2-KdV}
			\frac{{\rm{d}} {\bs{u}}_{0,0}}{{\rm{d}} y} =  m {\bs{u}}_{0,0} {\bs{g}}_{0,1} -(m-1) {\bs{g}}_{0,0}  {\bs{u}}_{0,0},~~~ \frac{{\rm{d}} \alpha}{{\rm{d}} y} = -1,
		\end{equation}
		respectively. 
	\end{proposition}
	
	A consequence of Propositions \ref{prop:KdV-sym-1} and \ref{prop:KdV-sym-2} is the following statement.
	\begin{proposition}
		The lowest order five-point generalised symmetries of equation \eqref{eq:nonAbel-KdV} is given by 
		\begin{equation} \label{eq:mast-sym-12-KdV}
			\frac{{\rm{d}} {\bs{u}}_{0,0}}{{\rm{d}} \epsilon} =   n {\bs{u}}_{0,0} {\bs{f}}_{1,0} -(n-1) {\bs{f}}_{0,0}  {\bs{u}}_{0,0}+  m {\bs{u}}_{0,0} {\bs{g}}_{0,1} -(m-1) {\bs{g}}_{0,0}  {\bs{u}}_{0,0},
		\end{equation}
		where ${\bs{f}}_{0,0} = ( {\bs{u}}_{0,0} {\bs{u}}_{-1,0}+\alpha)^{-1} $ and  ${\bs{g}}_{0,0} = ({\bs{u}}_{0,0} {\bs{u}}_{0,-1}-\alpha)^{-1} $.
	\end{proposition}
	
	Using the generalised symmetries we found, we can reduce equation \eqref{eq:nonAbel-KdV} to a map, following the procedure outlined in Section \ref{sec:red-KdV}. Since the KdV equation admits only a scaling symmetry, generated by $\tfrac{{\rm d}}{{\rm d}t} \bsu_{0,0} = (-1)^{n+m} \bsu_{0,0}$, the resulting reduced map will be expressed in terms of the invariants associated with this symmetry. These invariants are given by
	\begin{equation} \label{eq:inv-KdV}
	 {\bs{x}}_{0,0} = \bsu_{1,0} \bsu_{0,0}, \quad {\bs{X}}_{0,0} = \bsu_{0,0} \bsu_{1,0}, \quad {\bs{y}}_{0,0} = \bsu_{0,1} \bsu_{0,0}, \quad {\bs{Y}}_{0,0} = \bsu_{0,0} \bsu_{0,1},
	 \end{equation}
	and are pairwise related by transposition, i.e., $\tau(\bsx_{0,0}) = \bsX_{0,0}$ and $\tau(\bsy_{0,0}) = \bsY_{0,0}$. Consequently, the reduced map will be six-dimensional, with its defining equations similarly related by transposition.
	
	Indeed, consider the reduction under the symmetry 
\[
\tfrac{{\rm{d}} {\bs{u}}_{0,0}}{{\rm{d}} \epsilon} + \lambda (-1)^{n+m} \bsu_{0,0} = 0,
\]
see \eqref{eq:mast-sym-12-KdV}. Following the approach of  Section~\ref{sec:red-KdV}, we use invariants \eqref{eq:inv-KdV} and their compatibility conditions,
\[
\bsx_{0,1} \bsY_{0,0}^{-1} = \bsy_{1,0} \bsX_{0,0}^{-1}, \qquad \bsy_{0,0}^{-1} \bsX_{0,1} = \bsx_{0,0}^{-1} \bsY_{1,0},
\]
to express both the KdV equation and the symmetry constraint in terms of these invariants.

The KdV equation admits two equivalent invariant forms. The first one is given by
		\begin{eqnarray*}
			&& {\bs{x}}_{0,1} = -\alpha + {\bs{X}}_{0,0}^{-1} ({\bs{X}}_{0,0}+\alpha) {\bs{Y}}_{0,0}, ~~~~~{\bs{y}}_{1,0} = \alpha + {\bs{Y}}_{0,0}^{-1} ({\bs{Y}}_{0,0}-\alpha) {\bs{X}}_{0,0}, \\
			&&   {\bs{X}}_{0,1} =  -\alpha + {\bs{y}}_{0,0} ({\bs{x}}_{0,0}+\alpha) {\bs{x}}_{0,0}^{-1}, ~~~ {\bs{Y}}_{1,0} = \alpha + {\bs{x}}_{0,0} ({\bs{y}}_{0,0}-\alpha) {\bs{y}}_{0,0}^{-1},
		\end{eqnarray*}
	and the second invariant form is
		\begin{eqnarray*}
			&& {\bs{x}}_{0,0} = \alpha \left({\bs{X}}_{0,1}- {\bs{y}}_{0,0}+\alpha\right)^{-1} {\bs{y}}_{0,0}, ~~~{\bs{y}}_{1,0} = \alpha {\bs{x}}_{0,1} ({\bs{x}}_{0,1}- {\bs{Y}}_{0,0}+\alpha)^{-1}, \\
			&& {\bs{X}}_{0,0} = \alpha {\bs{Y}}_{0,0} ({\bs{x}}_{0,1}-{\bs{Y}}_{0,0}+\alpha)^{-1}, ~~~ {\bs{Y}}_{1,0} =\alpha ({\bs{X}}_{0,1}- {\bs{y}}_{0,0}+\alpha)^{-1} {\bs{X}}_{0,1}.
		\end{eqnarray*}
	Similarly, the invariant form of the symmetry constraint is
	\begin{equation*} 
		n ({\bs{x}}_{0,0}+\alpha)^{-1} - (n-1) ({\bs{X}}_{-1,0}+\alpha)^{-1} + m ({\bs{y}}_{0,0}-\alpha)^{-1} - (m-1) ({\bs{Y}}_{0,-1}-\alpha)^{-1} + \lambda (-1)^{n+m}= 0,
	\end{equation*}
	\begin{equation*} 
		n ({\bs{X}}_{0,0}+\alpha)^{-1} - (n-1) ({\bs{x}}_{-1,0}+\alpha)^{-1} + m ({\bs{Y}}_{0,0}-\alpha)^{-1} - (m-1) ({\bs{y}}_{0,-1}-\alpha)^{-1} + \lambda (-1)^{n+m}= 0.
	\end{equation*}
	
	Proceeding as in Section~\ref{sec:red-KdV}, we introduce the variables 
	\[
	\bsx_n := \bsx_{-1,0}, \quad \bsy_n := \bsy_{0,0}, \quad \bsz_n := \bsx_{0,0}, \quad 
	\bsX_n := \bsX_{-1,0}, \quad \bsY_n := \bsY_{0,0}, \quad \bsZ_n := \bsX_{0,0},
	\]
	in terms of which we obtain a six-dimensional map. Its first four equations have the following form,
	\begin{eqnarray*}
			&&{\bs{x}}_{n+1} = {\bs{z}}_n, ~~~  {\bs{y}}_{n+1} = \alpha + {\bs{Y}}_{n}^{-1} ({\bs{Y}}_{n}-\alpha) {\bs{Z}}_{n},  \\
			&& {\bs{X}}_{n+1} = {\bs{Z}}_n , ~~~ {\bs{Y}}_{n+1} = \alpha + {\bs{z}}_{n} ({\bs{y}}_{n}-\alpha) {\bs{y}}_{n}^{-1},
		\end{eqnarray*}
whereas the equations for the other two variables, ${\bs{z}}_{n+1}$ and ${\bs{Z}}_{n+1}$, are determined by the following relations.
	\begin{eqnarray*}
	&& \alpha (n+1) (\bsz_{n+1}+\alpha)^{-1} + (n-1) (\bsx_n+\alpha)^{-1} {\bs{Z}}_n - \lambda (-1)^{n+m} ({\bs{Z}}_n+ \alpha)  + \nonumber \\
	&& \hspace{5cm} m \left(\alpha {\bs{Z}}_n^{-1} (\bsY_n-\alpha)^{-1} \bsY_n - (\bsY_n - \alpha)^{-1} {\bs{Z}}_n\right) + m-n-1 =0,  \\
	&& \alpha (n+1) (\bsZ_{n+1}+\alpha)^{-1} + (n-1) {\bs{z}}_n (\bsX_n+\alpha)^{-1} - \lambda (-1)^{n+m} ({\bs{z}}_n+ \alpha)  + \nonumber \\
	&& \hspace{5cm} m \left(\alpha \bsy_n (\bsy_n-\alpha)^{-1} {\bs{z}}_n^{-1} - {\bs{z}}_n (\bsy_n - \alpha)^{-1} \right) + m-n-1 =0.
	\end{eqnarray*}	
	As anticipated, the map is six-dimensional and invariant under transposition. If variables were commutative, then $\bsX_n=\bsx_n$, $\bsY_n=\bsy_n$, $\bsZ_n=\bsz_n$ and the map reduces to a  three-dimensional one. Moreover, the commutative map could be decoupled to a third-order equation for either $y_n$ or $z_n$ which are omitted here because of their length.
	
	\subsection{Auto-B\"acklund transformation and Yang--Baxter map}
	
	Using the Lax pair of an equation, one can derive an auto-B\"acklund transformation via a Darboux transformation that leaves the Lax pair covariant. More precisely, consider the Lax pair associated with the equation ${\bs{Q}}(\bsu) = 0$,
	$$	{\bs{\Psi}}_{1,0} = \bs{L}_{0,0} \bs{\Psi}_{0,0} , ~~ \bs{\Psi}_{0,1} = \bs{M}_{0,0} \bs{\Psi}_{0,0},$$
	where the Lax matrices $\bs{L}$ and $\bs{M}$ depend on $\bsu$ and its shifts, and the spectral parameter $\lambda$. A Darboux transformation maps this Lax pair to a new one,
	$$	\tilde{\bs{\Psi}}_{1,0} = \tilde{\bs{L}}_{0,0} \tilde{\bs{\Psi}}_{0,0} , ~~ \tilde{\bs{\Psi}}_{0,1} = \tilde{\bs{M}}_{0,0} \tilde{\bs{\Psi}}_{0,0},$$
	where $\tilde{\bs{L}}$ and $\tilde{\bs{M}}$ are obtained from $\bs{L}$ and $\bs{M}$ by replacing the solution $\bsu$ with another solution $\tilde{\bsu}$, satisfying ${\bs{Q}}(\tilde{\bsu}) = 0$. 
	
	The fundamental solutions $\bs{\Psi}$ and $\tilde{\bs{\Psi}}$ are related through a Darboux transformation of the form
	$$	\tilde{{\bs{\Psi}}}_{0,0} = {\bs{D}}_{0,0} {\bs{\Psi}}_{0,0},$$
	where the Darboux matrix $\bs{D}$ generally depends on $\bsu$, $\tilde{\bsu}$, the spectral parameter $\lambda$, the B\"acklund parameter $\gamma$, and possibly an auxiliary function (potential) $\bsv$. As a consequence of this transformation, the Darboux matrix $\bs{D}$ must satisfy the compatibility conditions
	$$ {\bs{D}}_{1,0} {\bs{L}}_{0,,0} = \tilde{\bs{L}}_{0,0} {\bs{D}}_{0,0}, ~~~ {\bs{D}}_{0,1} {\bs{M}}_{0,0} = \tilde{\bs{M}}_{0,0} {\bs{D}}_{0,0},$$
	which in turn yield an auto-B\"acklund transformation for the original equation.

	For the KdV equation \eqref{eq:nonAbel-KdV} and its Lax pair \eqref{eq:Lax-KdV}, the Darboux transformation is given by
	\begin{equation}
		\tilde{{\bs{\Psi}}}_{0,0} = {\bs{D}}(\bsv_{0,0},\gamma) {\bs{\Psi}}_{0,0} = \begin{pmatrix}
			\gamma \bsv_{0,0}^{-1} & \lambda  \\ \lambda &  \bsv_{0,0} 	\end{pmatrix}   {\bs{\Psi}}_{0,0}.
	\end{equation}
	The first row of the consistency condition ${\bs{D}}(\bsv_{1,0},\gamma) {\bs{L}}_{0,0} = \tilde{\bs{L}}_{0,0} {\bs{D}}(\bsv_{0,0},\gamma)$ implies that 
	$$ \tilde{\bsu}_{0,0} = \bsv_{0,0}^{-1} \bsu_{0,0} \bsv_{1,0} ~~ {\mbox{and}} ~~ \bsv_{1,0} = \bsu_{0,0}^{-1} (\alpha \bsv_{0,0}-\gamma \bsu_{0,0}) (\bsu_{0,0}-\bsv_{0,0})^{-1},$$
	whereas the second row yields
	$$ \tilde{\bsu}_{0,0} = \bsv_{1,0} \bsu_{0,0} \bsv_{0,0}^{-1} ~~ {\mbox{and}} ~~ \bsv_{1,0} = (\bsu_{0,0}-\bsv_{0,0})^{-1} (\alpha \bsv_{0,0}-\gamma \bsu_{0,0}) \bsu_{0,0}^{-1}.$$
	Taking into account the identity 
	\begin{equation} \label{eq:identity}
		\bsx^{-1} \left(\alpha \bsx + \beta \bsy\right) (\gamma \bsx + \delta \bsy)^{-1} = (\gamma \bsx + \delta \bsy)^{-1} \left(\alpha \bsx + \beta \bsy\right) \bsx^{-1}, 
	\end{equation}
	which holds for any $\bsx, \bsy \in {\mathfrak{U}}$, the two expressions for $\bsv_{1,0}$ are equivalent, and as a consequence the same is true for the two expressions for $\tilde{\bsu}$. In view of this, we have 
	$$ \tilde{\bsu}_{0,0} =\bsv_{0,0}^{-1} \left(\alpha \bsv_{0,0}-\gamma \bsu_{0,0} \right) (\bsu_{0,0} - \bsv_{0,0})^{-1} =  (\bsu_{0,0} - \bsv_{0,0})^{-1} \left(\alpha \bsv_{0,0}-\gamma \bsu_{0,0} \right) \bsv_{0,0}^{-1}, $$
	and the equation for $\bsv_{1,0}$ can be written as
	$$ \bsv_{1,0} =  (\alpha  - \gamma ) (\bsu_{0,0}-\bsv_{0,0})^{-1} - \alpha \bsu_{0,0}^{-1}.$$	
	
	The second condition ${\bs{D}}(\bsv_{0,1},\gamma) {\bs{M}}_{0,0} = \tilde{\bs{M}}_{0,0} {\bs{D}}(\bsv_{0,0},\gamma)$ leads to an equation for $\bsv_{0,1}$. More precisely, its $(2,1)$ entry gives 
	$$\bsv_{0,1} = \bsu_{0,1} -\alpha \bsu_{0,0}^{-1} +\gamma \bsv_{0,0}^{-1},$$
	whereas its first row becomes an identity in view of the above  relations and after some computations. We can summarise this derivation in
	\begin{proposition} \label{prop:autoBT-HKdV}
		An auto-B\"acklund transformation for the KdV equation \eqref{eq:nonAbel-KdV} is given by 
		\begin{subequations} \label{eq:autoBT-HKdV}
			\begin{equation} \label{eq:autoBT-HKdV-u}
				\tilde{\bsu}_{0,0} =\bsv_{0,0}^{-1} \left(\alpha \bsv_{0,0}-\gamma \bsu_{0,0} \right) (\bsu_{0,0} - \bsv_{0,0})^{-1} =  (\bsu_{0,0} - \bsv_{0,0})^{-1} \left(\alpha \bsv_{0,0}-\gamma \bsu_{0,0} \right) \bsv_{0,0}^{-1}
			\end{equation}
			where the potential $\bsv$ is determined by the system
			\begin{equation}
				\bsv_{1,0} =  (\alpha  - \gamma ) (\bsu_{0,0}-\bsv_{0,0})^{-1} - \alpha \bsu_{0,0}^{-1}, ~~~ 
				\bsv_{0,1} = \bsu_{0,1} -\alpha \bsu_{0,0}^{-1} +\gamma \bsv_{0,0}^{-1}. \label{eq:autoBT-HKdV-v}
			\end{equation}
		\end{subequations}
	\end{proposition}

	\begin{center}
		\begin{figure}[h]
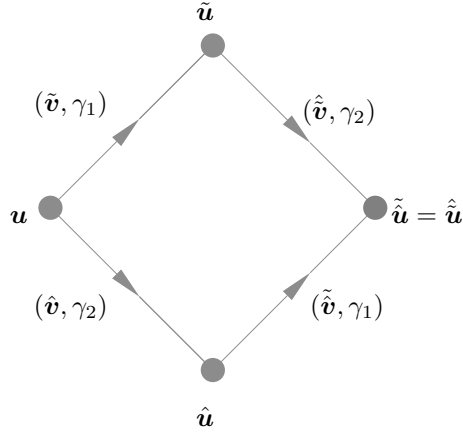

			\centertexdraw{ \setunitscale 0.85
				\linewd 0.002 \arrowheadtype t:F \arrowheadsize l:0.2 w:0.08
				\htext(0 0.5) {\phantom{T}}
				\setgray 0.55 \move (1 1) \avec (1.55 .45) \lvec(2 0) \avec (2.6 .6) \lvec (3 1) 
				\move (1 1) \avec(1.55 1.55 )\lvec(2 2) \avec(2.6 1.4) \lvec(3 1) 
				\move (1 1) \fcir f:0.55 r:0.075 \move (2 0) \fcir f:0.55 r:0.075
				\move (2 2) \fcir f:0.55 r:0.075 \move (3 1) \fcir f:0.5 r:0.075
				\htext(.75 .9) {$\bsu$} \htext(3.1 .9) {$\tilde{\hat{\bsu}}=\hat{\tilde{\bsu}}$}
				\htext(1.9 -.35) {$\hat{\bsu}$} \htext(1.9 2.15) {$\tilde{\bsu}$}
				\htext(0.9 .3) {$(\hat{\bsv},\gamma_2)$} \htext(0.9 1.55) {$(\tilde{\bsv},\gamma_1)$}
				\htext(2.6 .3) {$(\tilde{\hat{\bsv}},\gamma_1)$} \htext(2.55 1.5) {$(\hat{\tilde{\bsv}},\gamma_2)$}
			}
			\caption{Bianchi commuting diagram} \label{fig:Bianchi}
		\end{figure}
	\end{center}
	
	The superposition principle of the auto-B\"acklund transformation \eqref{eq:autoBT-HKdV} follows by the permutation of four Darboux matrices  according to the Bianchi commuting diagram in Figure \ref{fig:Bianchi},
	\begin{equation} \label{eq:Lax-H3B} 
		{\bs{D}}(\hat{\tilde{\bsv}},\gamma_2) {\bs{D}}(\tilde{\bsv},\gamma_1)  = {\bs{D}}(\tilde{\hat{\bsv}},\gamma_1) {\bs{D}}(\hat{\bsv},\gamma_2),
	\end{equation}
	leading to 
	\begin{equation} \label{eq:pot-v-sp}
		\hat{\tilde{\bsv}} = \tilde{\bsv}^{-1}\left(\gamma_1 \hat{\bsv} -  \gamma_2 \tilde{\bsv} \right) \left(\tilde{\bsv} -  \hat{\bsv} \right)^{-1}, ~~ \tilde{\hat{\bsv}} = \hat{\bsv}^{-1}\left(\gamma_1 \hat{\bsv} -  \gamma_2 \tilde{\bsv} \right) \left(\tilde{\bsv} -\hat{\bsv} \right)^{-1}.
	\end{equation}
	The new solution of \eqref{eq:nonAbel-KdV} is given by 
	\begin{subequations} \label{eq:autoBT-HKdV-sp}
		\begin{equation}
			\hat{\tilde{\bsu}} =  \left( \tilde{\bsv} -\hat{\bsv} \right) \left(\gamma_1 \hat{\bsv}  -  \gamma_2 \tilde{\bsv}  \right)^{-1}  \tilde{\bsv}  (\tilde{\bsv}-\hat{\bsv})^{-1} \hat{\bsv} {\bs{A}} \bsu {\bs{R}}^{-1} \left( \tilde{\bsv} -\hat{\bsv} \right),
		\end{equation}
		where
		\begin{eqnarray}
			{\bs{A}} &:=& \gamma_1 (\alpha -\gamma_2) \tilde{\bsv}^{-1}+\gamma_2 (\gamma_1-\alpha) \hat{\bsv}^{-1} + \alpha (\gamma_2-\gamma_1)  \bsu^{-1},\\
			{\bs{R}} &:=&  (\alpha-\gamma_2) \tilde{\bsv} + (\gamma_1-\alpha) \hat{\bsv} + (\gamma_2-\gamma_1) \bsu,
		\end{eqnarray}
	\end{subequations}
	and its derivation is given in the Appendix. It can be readily verified that superposition principle formula \eqref{eq:autoBT-HKdV-sp} is invariant under the interchanges $\tilde{\bsv} \leftrightarrow \hat{\bsv}$ and $\gamma_1 \leftrightarrow \gamma_2$. Assuming all variables commute, the transformation \eqref{eq:autoBT-HKdV} and its superposition \eqref{eq:autoBT-HKdV-sp} reduce to the auto-Bäcklund transformation of the commutative KdV equation and its corresponding superposition principle, respectively.
	
	\begin{remark}
		Equation \eqref{eq:Lax-H3B} is the Lax representation of a Yang-Baxter map. To make contact with the standard notation of Yang-Baxter maps, let us set
		$$\bsy = {\hat{\tilde{\bsv}}}^{-1}, ~~ \bsx = {\tilde{\bsv}}^{-1}, ~~ {\bsy^\prime} = \hat{\bsv}^{-1}, ~~ {\bsx^\prime} = {\tilde{\hat{\bsv}}}^{-1}.$$
		In this notation, equation \eqref{eq:Lax-H3B} becomes 
		$$ 	{\bs{D}}(\bsy^{-1},\gamma_2) {\bs{D}}(\bsx^{-1},\gamma_1)  = {\bs{D}}({\bsx^\prime}^{-1},\gamma_1) {\bs{D}}({\bsy^\prime}^{-1},\gamma_2),$$
		and yields the map
		\begin{subequations}
			\begin{eqnarray}
				\bsx^\prime &=& \bsy \left(1 + \gamma_2 \bsx \bsy \right)  \left(1 + \gamma_1 \bsx \bsy \right)^{-1} =   \left(1 + \gamma_2 \bsy \bsx \right) \left(1 + \gamma_1 \bsy \bsx \right)^{-1}  \bsy \\
				\bsy^\prime &=&  \left(1 + \gamma_1 \bsx \bsy \right)  \left(1 + \gamma_2 \bsx \bsy \right)^{-1} \bsx  =   \bsx   \left(1 + \gamma_1 \bsy \bsx \right) \left(1 + \gamma_2 \bsy \bsx \right)^{-1},
			\end{eqnarray}
		\end{subequations}
		which is the non-commutative Yang--Baxter $F_{III}$ map given in \cite{Dol}. Finally, using the two equivalent forms of the map, one can easily show that ${\bs{I}} = \bsx \bsy + \bsy \bsx$ is an invariant of the map. 
	\end{remark}

	\section{Conclusions} \label{sec:con}
	
	We presented a method to find generalised symmetries of noncommutative difference equations by employing their Lax pair using the discrete potential KdV equation \eqref{eq:nonAbel-pKdV} as an illustrative example. This approach was also used to find the symmetries of the discrete KdV equation in Section \ref{sec:KdV}, and can be readily applied to find symmetries for other systems admitting a $2 \times 2$ Lax pair, such as the noncommutative Schr\"odinger system \cite{KRX}
	\begin{equation} \label{eq:NLS}
		{\bs{u}}_{1,0}	- \bsu_{0,1} =  (\alpha-\beta) (\bsu_{0,0} \bsv_{1,1}+1)^{-1} \bsu_{0,0}, ~~~ {\bs{v}}_{1,0}	- \bsv_{0,1} = (\beta-\alpha) \bsv_{1,1} (\bsu_{0,0} \bsv_{1,1}+1)^{-1}.
	\end{equation}
	Indeed, this system admits the Lax pair
	\begin{equation} \label{eq:Lax-NLS}
		{\bs{\Psi}}_{1,0} = \begin{pmatrix}
			\lambda + \alpha + \bsu_{0,0} \bsv_{1,0} & \bsu_{0,0} \\ \bsv_{1,0} & 1
		\end{pmatrix}{\bs{\Psi}}_{0,0} , ~~~ 	{\bs{\Psi}}_{0,1} = \begin{pmatrix}
			\lambda + \beta + \bsu_{0,0} \bsv_{0,1} & \bsu_{0,0} \\ \bsv_{0,1} & 1
		\end{pmatrix}{\bs{\Psi}}_{0,0},
	\end{equation}
	and it can be shown that
	$$ 	\frac{{\rm{d}} {\bs{\Psi}}_{0,0}}{{\rm{d}} t_1} = {\bs{A}}_{0,0} {\bs{\Psi}}_{0,0}, ~~~ \frac{{\rm{d}} {\bs{\Psi}}_{0,0}}{{\rm{d}} x} =n {\bs{A}}_{0,0} {\bs{\Psi}}_{0,0} , ~~~ \frac{{\rm{d}} \alpha}{{\rm{d}} x} = -1,$$
	where
	$$ {\bs{A}}_{0,0} = \frac{1}{\alpha+\lambda} \begin{pmatrix}
		\bsf_{0,0}  & - \bsf_{0,0} \bsu_{-1,0} \\ 
		-\bsv_{1,0} \bsf_{0,0} &  \bsv_{1,0} \bsf_{0,0} \bsu_{-1,0}
	\end{pmatrix} ~~ {\text{with}} ~~\bsf_{0,0} = (\bsu_{-1,0} \bsv_{1,0}+1)^{-1}, $$
	are consistent with the first equation in \eqref{eq:Lax-NLS} provided that
	$$ 	\frac{{\rm{d}} {\bs{u}}_{0,0}}{{\rm{d}} t_1} = \bsf_{0,0} \bsu_{-1,0} , ~~~  \frac{{\rm{d}} {\bs{v}}_{0,0}}{{\rm{d}} t_1} = -\bsv_{1,0} \bsf_{0,0}, $$
	and
	$$ 	\frac{{\rm{d}} {\bs{u}}_{0,0}}{{\rm{d}} x} = n \bsf_{0,0} \bsu_{-1,0} , ~~~ 	\frac{{\rm{d}} {\bs{v}}_{0,0}}{{\rm{d}} x} = -n \bsv_{1,0} \bsf_{0,0}, ~~~ \frac{{\rm{d}} \alpha}{{\rm{d}} x} = -1, $$
	respectively. The above two lattices generate the lowest order generalised symmetries of system \eqref{eq:NLS} in the first direction. Using arguments about the invariance of \eqref{eq:NLS} under the interchange of indices and parameters, one could find the symmetries of \eqref{eq:NLS} in the other direction.
	
	The Lax pair and Darboux transformation also provide a means to derive an auto-B\"acklund transformation for the equation under consideration. Such a transformation may depend solely on the variables $\bsu$ and $\tilde{\bsu}$, as is the case for the discrete potential KdV equation and its transformation \eqref{eq:BT-dpKdV}, as well as the corresponding transformation of system \eqref{eq:NLS} discussed in \cite{KRX}. However, an auto-B\"acklund transformation may also depend on an auxiliary potential as in \eqref{eq:autoBT-HKdV}. In such cases, its superposition principle leads to a Yang--Baxter map rather than a quad system, as demonstrated in Section \ref{sec:KdV}.
	
	We also demonstrated how generalised symmetries can be employed to reduce a noncommutative difference equation to a corresponding map, which may be related to a discrete Painlevé equation. A key distinction between the commutative and noncommutative cases lies in the dimension of the reduced map: in the noncommutative setting, this dimension may be higher due to the increased number of invariants that must be introduced as we explained in Section \ref{sec:KdV}. Additionally, we presented a systematic derivation of a system of noncommutative partial differential equations arising from reductions of equation \eqref{eq:nonAbel-pKdV} which generalise very well-known systems such as the Ernst equation and the Neugebauer--Kramer involution.
	
	In this work, we have focused exclusively on systems admitting a $2 \times 2$ Lax pair. It would be of interest to extend these considerations and derivations to equations that admit higher-order Lax pairs. One such equation is
	$$\bsu_{1,1} \left(\bsu_{1,0}+ \bsu_{0,1}\right) \bsu_{0,0} + 1 = 0,$$
	whose commutative counterpart was derived in \cite{MX} in the context of second-order integrability conditions. The associated Lax pair involves $3 \times 3$ matrices and is given by
	$${\bs{\Psi}}_{1,0} = \begin{pmatrix}
		0 &  \bsu_{1,0}^{-1} & \lambda \\ \lambda & 0 & \bsu_{0,0}^{-1} \\ -2 \bsu_{1,0} \bsu_{0,0} & \lambda & 0 
	\end{pmatrix}{\bs{\Psi}}_{0,0} , ~~~ 	{\bs{\Psi}}_{0,1} = \begin{pmatrix}
		0 &  -\bsu_{0,1}^{-1} & \lambda \\ \lambda & 0 & -\bsu_{0,0}^{-1} \\ 2 \bsu_{0,1} \bsu_{0,0} & \lambda & 0 
	\end{pmatrix}{\bs{\Psi}}_{0,0}.$$
	
	Symmetries of a commutative system of difference equation can be derived systematically either from first principles (see, for instance, \cite{FX,TTX2}) or via the theory of integrability conditions (see, for example, \cite{BX,LY,MWX,MX,X}). The latter has recently been extended to the study of noncommutative differential-difference equations in \cite{NW}, where a related classification problem was also addressed. It would be of interest to further develop the theory of integrability conditions for noncommutative difference equations, in analogy with the commutative case \cite{MWX,MX}.

	\section*{Appendix}
	
For completeness, we present in this appendix the derivation of system \eqref{eq:rpde}, given in the context of continuous symmetric reductions discussed in Section \ref{sec:rpde}; the proofs of two propositions stated in Section \ref{sec:KdV}; and the derivation of the superposition principle \eqref{eq:autoBT-HKdV-sp}.
	
	\subsubsection*{Derivation of system \eqref{eq:rpde}}
	
	We want to  eliminate variables $\bsu_{1,1}$, $\bsu_{-1,0}$ and $\bsu_{0,-1}$, and derive a system for $\bsu_{0,0}$, $\bsu_{1,0}$ and $\bsu_{0,1}$, using equations \eqref{eq:nonAbel-pKdV}, \eqref{eq:con-const} and their shifts. We start by shifting in $m$ the first equation in \eqref{eq:con-const} and using relation \eqref{eq:p-subs-b}. 
	\begin{eqnarray*}
		\frac{\partial {\bs{u}}_{0,1}}{\partial \alpha} &=&  -n \left({\bs{u}}_{1,1} -{\bs{u}}_{-1,1}\right)^{-1} = -\frac{n}{\alpha-\beta} \left({\bs{u}}_{1,0} -{\bs{u}}_{0,1}\right) \left({\bs{u}}_{1,0} -{\bs{u}}_{-1,0}\right)^{-1} \left({\bs{u}}_{-1,0} -{\bs{u}}_{0,1}\right)\\
		&=& -\frac{n}{\alpha-\beta} \left({\bs{u}}_{1,0} -{\bs{u}}_{0,1}\right)\left({\bs{u}}_{1,0} -{\bs{u}}_{-1,0}\right)^{-1}  \left( {\bs{u}}_{-1,0} -  {\bs{u}}_{1,0}+ {\bs{u}}_{1,0} -{\bs{u}}_{0,1} \right) \\
		&=&\frac{1}{\alpha-\beta} \left({\bs{u}}_{1,0} -{\bs{u}}_{0,1}\right) \left(n+ \frac{\partial {\bs{u}}_{0,0}}{\partial \alpha}  \left( {\bs{u}}_{1,0} -{\bs{u}}_{0,1} \right)\right).
	\end{eqnarray*} 
	Working in a similar fashion, we shift in $n$ the second equation in \eqref{eq:con-const} to find that
	\begin{equation} \label{eq:con-u1b}
		\frac{\partial {\bs{u}}_{1,0}}{\partial \beta} =\frac{1}{\alpha-\beta} \left({\bs{u}}_{1,0} -{\bs{u}}_{0,1}\right) \left(m - \frac{\partial {\bs{u}}_{0,0}}{\partial \beta}  \left( {\bs{u}}_{1,0} -{\bs{u}}_{0,1} \right)\right).
	\end{equation}
	Finally, we differentiate the first equation in \eqref{eq:con-const} in $\beta$,
	\begin{equation} \label{eq:con-uab}
		\frac{\partial^2 {\bs{u}}_{0,0}}{\partial \alpha \partial \beta} = n \left({\bs{u}}_{1,0} -{\bs{u}}_{-1,0}\right)^{-1} \left( \frac{\partial {\bs{u}}_{1,0}}{\partial \beta}  - \frac{\partial {\bs{u}}_{-1,0}}{\partial \beta}  \right) \left({\bs{u}}_{1,0} -{\bs{u}}_{-1,0}\right)^{-1},
	\end{equation}
	and use \eqref{eq:con-u1b} along with its backward shift, 
	$$	\frac{\partial {\bs{u}}_{-1,0}}{\partial \beta} =\frac{1}{\alpha-\beta}  \left(m -   \left( {\bs{u}}_{-1,0} -{\bs{u}}_{0,1} \right) \frac{\partial {\bs{u}}_{0,0}}{\partial \beta} \right) \left( {\bs{u}}_{-1,0} -{\bs{u}}_{0,1} \right) , $$
	to replace the derivatives with respect to $\beta$. Setting ${\bs{\Delta}} := {\bs{u}}_{1,0} -{\bs{u}}_{0,1}$ and ${\bs{\delta}} :=  {\bs{u}}_{-1,0} -{\bs{u}}_{0,1}$, we have that
	\begin{eqnarray*}
		\frac{\partial {\bs{u}}_{1,0}}{\partial \beta}  - \frac{\partial {\bs{u}}_{-1,0}}{\partial \beta}  &=& \frac{1}{\alpha-\beta} \left(  m \left({\bs{\Delta}}-{\bs{\delta}}\right) - {\bs{\Delta}} \frac{\partial {\bs{u}}_{0,0}}{\partial \beta}  {\bs{\Delta}}  +
		{\bs{\delta}} \frac{\partial {\bs{u}}_{0,0}}{\partial \beta}  {\bs{\delta}}\right) \\
		&=& \frac{1}{\alpha-\beta} \left(  m \left({\bs{\Delta}}-{\bs{\delta}}\right) - \left({\bs{\Delta}}-{\bs{\delta}}\right) \frac{\partial {\bs{u}}_{0,0}}{\partial \beta}  {\bs{\Delta}} - {\bs{\delta}} \frac{\partial {\bs{u}}_{0,0}}{\partial \beta} \left({\bs{\Delta}}- {\bs{\delta}}\right) \right)  \\
		&=&  \frac{1}{\alpha-\beta} \left(  m \left({\bs{\Delta}}-{\bs{\delta}}\right) - \left({\bs{\Delta}}-{\bs{\delta}}\right) \frac{\partial {\bs{u}}_{0,0}}{\partial \beta}  {\bs{\Delta}} -  \left(  {\bs{u}}_{-1,0} -{\bs{u}}_{1,0} + {\bs{\Delta}}\right)  \frac{\partial {\bs{u}}_{0,0}}{\partial \beta} \left({\bs{\Delta}}- {\bs{\delta}}\right) \right).
	\end{eqnarray*}
	Taking into account the above relation and that ${\bs{\Delta}}- {\bs{\delta}} = \bsu_{1,0}-\bsu_{-1,0}$, equation \eqref{eq:con-uab} becomes
	\begin{eqnarray*} 
		\frac{\partial^2 {\bs{u}}_{0,0}}{\partial \alpha \partial \beta} = \frac{1}{\alpha-\beta} \left(  \frac{\partial {\bs{u}}_{0,0}}{\partial \alpha} \left({\bs{u}}_{1,0} -{\bs{u}}_{0,1}\right)  \frac{\partial {\bs{u}}_{0,0}}{\partial \beta} + \frac{\partial {\bs{u}}_{0,0}}{\partial \beta} \left({\bs{u}}_{1,0} -{\bs{u}}_{0,1}\right)  \frac{\partial {\bs{u}}_{0,0}}{\partial \alpha} + n \frac{\partial {\bs{u}}_{0,0}}{\partial \beta} -  m \frac{\partial {\bs{u}}_{0,0}}{\partial \alpha} \right).
	\end{eqnarray*}
	
	\subsubsection*{Proof of Proposition \ref{prop:KdV-sym-1}}
	Differentiating \eqref{eq:nonAbel-KdV} with respect to $t_1$, we get
	$$\bsF_{0,0} - \alpha \bsu_{1,0}^{-1} \bsF_{1,0} \bsu_{1,0}^{-1} + \alpha \bsu_{0,1}^{-1} \bsF_{0,1} \bsu_{0,1}^{-1}- \bsF_{1,1} = 0,$$
	where $\bsF_{0,0}=  {\bs{u}}_{0,0} {\bs{f}}_{1,0} -{\bs{f}}_{0,0}  {\bs{u}}_{0,0}$ and $\bsf_{0,0}$ is given in \eqref{eq:Lax-difdif-KdV-2}. Replacing $\bsF$ and its shifts, we arrive at
	$$(\bsu_{0,0}+\alpha \bsu_{1,0}^{-1}) \bsf_{1,0} + \bsf_{1,1} (\alpha \bsu_{0,1}^{-1}+\bsu_{1,1}) = \bsf_{0,0} \bsu_{0,0} + \alpha \bsf_{2,0} \bsu_{1,0}^{-1} +\alpha \bsu_{0,1}^{-1} \bsf_{0,1}+\bsu_{1,1} \bsf_{2,1}.$$
	Since $\bsu_{0,0}+\alpha \bsu_{1,0}^{-1} = \bsu_{1,0}^{-1} \bsf_{1,0}^{-1}$ and 
	$\alpha \bsu_{0,1}^{-1}+\bsu_{1,1} = \bsf_{1,1}^{-1} \bsu_{0,1}^{-1}$, the above relation becomes
	\begin{equation} \label{eq:rel1}
		\bsu_{1,0}^{-1}+ \bsu_{0,1}^{-1}= \bsf_{0,0} \bsu_{0,0} + \alpha \bsf_{2,0} \bsu_{1,0}^{-1} +\alpha \bsu_{0,1}^{-1} \bsf_{0,1}+\bsu_{1,1} \bsf_{2,1}.
	\end{equation}
	From the forward and backward shifts of \eqref{eq:nonAbel-KdV} we find that
	\begin{subequations} \label{eq:subs}
		\begin{equation}
			{\bs{u}}_{2,0}^{-1} ({\bs{u}}_{2,0} {\bs{u}}_{1,0} + \alpha) = (\alpha + \bsu_{2,1} \bsu_{1,1}){\bs{u}}_{1,1}^{-1} ~ \Rightarrow ~ \bsf_{2,0} \bsu_{2,0} = \bsu_{1,1} \bsf_{2,1},
		\end{equation}
		and
		\begin{equation}
			{\bs{u}}_{0,0}^{-1} ({\bs{u}}_{0,0} {\bs{u}}_{-1,0} + \alpha) = (\alpha + \bsu_{0,1} \bsu_{-1,1}){\bs{u}}_{-1,1}^{-1}  ~ \Rightarrow ~ \bsf_{0,0} \bsu_{0,0} = \bsu_{-1,1} \bsf_{0,1},
		\end{equation}
	\end{subequations}
	respectively. In view of \eqref{eq:subs}, equation \eqref{eq:rel1} becomes
	\begin{eqnarray*} 
		\bsu_{1,0}^{-1}+ \bsu_{0,1}^{-1} &=&  \bsu_{-1,1} \bsf_{0,1} + \alpha \bsf_{2,0} \bsu_{1,0}^{-1} +\alpha \bsu_{0,1}^{-1}  \bsf_{0,1} + \bsf_{2,0} \bsu_{2,0} \nonumber \\
		&=& \bsf_{2,0} (\bsu_{2,0} \bsu_{1,0} + \alpha) \bsu_{1,0}^{-1} + \bsu_{0,1}^{-1}(  \bsu_{0,1}   \bsu_{-1,1} + \alpha) \bsf_{0,1}\nonumber \\
		&=& \bsu_{1,0}^{-1}+ \bsu_{0,1}^{-1}.
	\end{eqnarray*}
	
	To show that the second equation in \eqref{eq:difdif-KdV-1} is another symmetry of \eqref{eq:nonAbel-KdV}, first we rewrite the former as
	\begin{equation}\label{eq:mast-sym-KdV-v2}
		\frac{{\rm{d}} {\bs{u}}_{0,0}}{{\rm{d}} x} =  n \frac{{\rm{d}} {\bs{u}}_{0,0}}{{\rm{d}} t_1}  + {\bs{f}}_{0,0}  {\bs{u}}_{0,0}, ~~~ 	\frac{{\rm{d}} \alpha}{{\rm{d}} x} = 1,
	\end{equation}
	and then differentiate the latter with respect to $x$. This leads to 
	$$ n \frac{{\rm{d}} {\bs{Q}}}{{\rm{d}} t_1} + \bsf_{0,0} \bsu_{0,0} + \bsu_{1,0}^{-1} - \alpha \bsu_{1,0}^{-1}\left(\frac{{\rm{d}} \bsu_{1,0}}{{\rm{d}} t_1} + \bsf_{1,0} \bsu_{1,0}\right) \bsu_{1,0}^{-1} - \bsu_{0,1}^{-1} + \alpha \bsu_{0,1}^{-1} \bsf_{0,1} - \frac{{\rm{d}} \bsu_{1,1}}{{\rm{d}} t_1} - \bsf_{1,1} \bsu_{1,1} = 0, $$
	where ${\bs{Q}}$ denotes the left hand side of \eqref{eq:nonAbel-KdV}. After taking into account the KdV equation and that $ \frac{{\rm{d}} {\bs{Q}}}{{\rm{d}} t_1} =0$, we end up with
	$$  \bsf_{0,0} \bsu_{0,0} + \bsu_{1,0}^{-1} - \alpha \bsf_{2,0} \bsu_{1,0}^{-1} - \bsu_{0,1}^{-1} + \alpha \bsu_{0,1}^{-1} \bsf_{0,1}  - \bsu_{1,1} \bsf_{2,1} = 0.$$
	Finally, we substitute $\bsf_{0,0}$ and $\bsf_{2,1}$ using relations \eqref{eq:subs} and the above relation becomes an identity.

	\subsubsection*{Proof of Proposition \ref{prop:Miura-KdV}}
	
	We rewrite \eqref{eq:Miura} as $\bsv_{0,0}^{-1} = \bsu_{1,0} + \alpha \bsu_{0,0}^{-1}$ and then we differentiate it in $t_1$ using the first symmetry in \eqref{eq:difdif-KdV-1} and its shift.
	\begin{eqnarray*}
		-\bsv_{0,0}^{-1} \frac{{\rm{d}} \bsv_{0,0}}{{\rm{d}} t_1} \bsv_{0,0}^{-1} &=&  \frac{{\rm{d}} {\bs{u}}_{1,0}}{{\rm{d}} t_1} - \alpha \bsu_{0,0}^{-1} \frac{{\rm{d}} {\bs{u}}_{0,0}}{{\rm{d}} t_1} \bsu_{0,0}^{-1}\\
		&=& \bsu_{1,0} \bsf_{2,0} - \bsf_{1,0} \bsu_{1,0} - \alpha \bsu_{0,0}^{-1} \left(\bsu_{0,0} \bsf_{1,0} - \bsf_{0,0} \bsu_{0,0}\right)  \bsu_{0,0}^{-1}  \\
		&=& \bsv_{1,0} - \bsf_{1,0} \left(\bsu_{1,0}\bsu_{0,0}+ \alpha\right) \bsu_{0,0}^{-1}   + \alpha \bsu_{0,0}^{-1} \bsf_{0,0} \\
		&=& \bsv_{1,0} -  \bsu_{0,0}^{-1}   + \alpha \bsu_{0,0}^{-1} \bsf_{0,0}\\ 
		&=& \bsv_{1,0} -  \bsu_{0,0}^{-1}\left( \bsf_{0,0}^{-1} - \alpha\right)\bsf_{0,0}\\
		&=& \bsv_{1,0} -  \bsu_{-1,0}\bsf_{0,0} \\
		&=& \bsv_{1,0} -  \bsv_{-1,0},
	\end{eqnarray*}
	where ${\bs{f}}_{0,0} = ( {\bs{u}}_{0,0} {\bs{u}}_{-1,0}+\alpha)^{-1}$. This is modified Volterra lattice \eqref{eq:Volt1} up to the change of $t_1$ to $-t_1$.
	
	Differentiating now $\bsv_{0,0}^{-1} = \bsu_{1,0} + \alpha \bsu_{0,0}^{-1}$ in $x$ and taking into account \eqref{eq:mast-sym-KdV-v2}, we arrive at
	\begin{eqnarray*}
		-\bsv_{0,0}^{-1} \frac{{\rm{d}} \bsv_{0,0}}{{\rm{d}} x} \bsv_{0,0}^{-1} &=& (n+1)  \frac{{\rm{d}} {\bs{u}}_{1,0}}{{\rm{d}} t_1} + \bsf_{1,0} \bsu_{1,0}  - \alpha \bsu_{0,0}^{-1} \left( n \frac{{\rm{d}} {\bs{u}}_{0,0}}{{\rm{d}} t_1} + \bsf_{0,0} \bsu_{0,0} \right)\bsu_{0,0}^{-1} + \bsu_{0,0}^{-1}   \\
		&=& n \left( \frac{{\rm{d}} {\bs{u}}_{1,0}}{{\rm{d}} t_1} - \alpha \bsu_{0,0}^{-1} \frac{{\rm{d}} {\bs{u}}_{0,0}}{{\rm{d}} t_1} \bsu_{0,0}^{-1} \right) + \frac{{\rm{d}} {\bs{u}}_{1,0}}{{\rm{d}} t_1} + \bsf_{1,0} \bsu_{1,0} - \alpha \bsu_{0,0}^{-1} \bsf_{0,0}  + \bsu_{0,0}^{-1}   \\
		&=& n ( \bsv_{1,0} -  \bsv_{-1,0}) +  {\bs{u}}_{1,0} {\bs{f}}_{2,0} -{\bs{f}}_{1,0}  {\bs{u}}_{1,0}  + \bsf_{1,0} \bsu_{1,0} - \alpha \bsu_{0,0}^{-1} \bsf_{0,0}  + \bsu_{0,0}^{-1} \\
		&=&n ( \bsv_{1,0} -  \bsv_{-1,0}) +  {\bs{v}}_{1,0}  - \bsu_{0,0}^{-1} \left( \alpha   - \bsf_{0,0}^{-1}\right) \bsf_{0,0}\\ 
		&=&n ( \bsv_{1,0} -  \bsv_{-1,0}) +  {\bs{v}}_{1,0} + \bsu_{-1,0} \bsf_{0,0}\\ 
		&=&n ( \bsv_{1,0} -  \bsv_{-1,0}) +  {\bs{v}}_{1,0}  + \bsv_{-1,0}\\ 
		&=& (n+1) \bsv_{1,0}  - (n-1) \bsv_{-1,0},
	\end{eqnarray*}
	which is the non-autonomous modified Volterra lattice \eqref{eq:mast-sym-mod-Volt} up to the change of $x$ to $-x$.
	
	\subsubsection*{Derivation of superposition principle \eqref{eq:autoBT-HKdV-sp}}
	
	According to the Bianchi diagram in Figure~\ref{fig:Bianchi}, we begin with a solution $\bsu$ and apply the transformation \eqref{eq:autoBT-HKdV} with parameter $\gamma_1$ to obtain
	\begin{equation} \label{eq:proof-sp-3}
		\tilde{\bsu} = \tilde{\bsv}^{-1} \left( \alpha \tilde{\bsv} - \gamma_1 \bsu \right) \left(\bsu - \tilde{\bsv}\right)^{-1}.
	\end{equation}
	Next, using $\tilde{\bsu}$ as the seed and applying \eqref{eq:autoBT-HKdV} with parameter $\gamma_2$, as indicated in the Bianchi diagram, we derive
	\begin{equation} \label{eq:proof-sp-1}
		\hat{\tilde{\bsu}} =  \hat{\tilde{\bsv}}^{-1} \left( \alpha \hat{\tilde{\bsv}} - \gamma_2 \tilde{\bsu} \right) \left(\tilde{\bsu} - \hat{\tilde{\bsv}} \right)^{-1},
	\end{equation}
	where $\hat{\tilde{\bsv}}$ is defined by \eqref{eq:pot-v-sp}, and can also be expressed using identity \eqref{eq:identity} as
	\begin{equation} \label{eq:proof-sp-2}
	\hat{\tilde{\bsv}} = \left(\tilde{\bsv} -  \hat{\bsv} \right)^{-1}\left(\gamma_1 \hat{\bsv} -  \gamma_2 \tilde{\bsv} \right) \tilde{\bsv}^{-1}.
	\end{equation}
	
	To express $\hat{\tilde{\bsu}}$ explicitly in terms of $\tilde{\bsv}$, $\hat{\bsv}$, and $\bsu$, we substitute the inverse of  \eqref{eq:pot-v-sp} into the first term of \eqref{eq:proof-sp-1}, i.e.,
	$$ \hat{\tilde{\bsv}}^{-1} =  \left(\tilde{\bsv} -  \hat{\bsv} \right) \left(\gamma_1 \hat{\bsv} -  \gamma_2 \tilde{\bsv} \right)^{-1} \tilde{\bsv},$$
	and use \eqref{eq:proof-sp-2} to replace $\hat{\tilde{\bsv}}$ in the remaining terms of \eqref{eq:proof-sp-1}. In particular, we compute
	\begin{eqnarray*}
		\alpha \hat{\tilde{\bsv}} - \gamma_2 \tilde{\bsu} &=& \left(\tilde{\bsv} -  \hat{\bsv} \right)^{-1} \left\{ \alpha \left(\gamma_1 \hat{\bsv} -  \gamma_2 \tilde{\bsv} \right) \tilde{\bsv}^{-1} \left(\bsu - \tilde{\bsv}\right) - \gamma_2   \left(\tilde{\bsv} -  \hat{\bsv} \right) \tilde{\bsv}^{-1} \left( \alpha \tilde{\bsv} - \gamma_1 \bsu \right)  \right\} \left(\bsu - \tilde{\bsv}\right)^{-1} \\
		&=& \left(\tilde{\bsv} -  \hat{\bsv} \right)^{-1} \left\{ \gamma_1 (\alpha-\gamma_2) \hat{\bsv} \tilde{\bsv}^{-1} \bsu -\alpha (\gamma_1-\gamma_2) \hat{\bsv} + \gamma_2 (\gamma_1-\alpha) \bsu \right\} \left(\bsu - \tilde{\bsv}\right)^{-1} \\
		&=&  \left(\tilde{\bsv} -  \hat{\bsv} \right)^{-1} \hat{\bsv}  \left\{ \gamma_1 (\alpha -\gamma_2) \tilde{\bsv}^{-1}+\gamma_2 (\gamma_1-\alpha) \hat{\bsv}^{-1} + \alpha (\gamma_2-\gamma_1)  \bsu^{-1} \right\} \bsu \left(\bsu - \tilde{\bsv}\right)^{-1} ,
	\end{eqnarray*}
	and
	\begin{eqnarray*}
		\tilde{\bsu}-\hat{\tilde{\bsv}} &=& \left(\tilde{\bsv} -  \hat{\bsv} \right)^{-1} \left\{ \left(\tilde{\bsv} -  \hat{\bsv} \right) \tilde{\bsv}^{-1} \left( \alpha \tilde{\bsv} - \gamma_1 \bsu \right) - \left(\gamma_1 \hat{\bsv} -  \gamma_2 \tilde{\bsv} \right) \tilde{\bsv}^{-1} \left(\bsu - \tilde{\bsv}\right) -   \right\} \left(\bsu - \tilde{\bsv}\right)^{-1} \\
		&=& \left(\tilde{\bsv} -  \hat{\bsv} \right)^{-1} \left\{    (\alpha-\gamma_2) \tilde{\bsv} + (\gamma_1-\alpha) \hat{\bsv} + (\gamma_2-\gamma_1) \bsu \right\} \left(\bsu - \tilde{\bsv}\right)^{-1}.
	\end{eqnarray*}
	Substituting these expressions into \eqref{eq:proof-sp-1}, we recover the formula \eqref{eq:autoBT-HKdV-sp}, thus completing the derivation.
	
	If instead we follow the alternative path $\bsu \rightarrow \hat{\bsu} \rightarrow \tilde{\hat{\bsu}}$ in the Bianchi diagram and apply the same procedure, we again arrive at \eqref{eq:autoBT-HKdV-sp}. This confirms that $\hat{\tilde{\bsu}} = \tilde{\hat{\bsu}}$, and hence the commutativity of the Bianchi diagram.

\end{document}